\documentclass[a4paper,twocolumn,10pt,unpublished]{quantumarticle}
\pdfoutput=1
\usepackage{graphicx}
\usepackage[dvipsnames]{xcolor}
\usepackage{amsmath}
\usepackage{amssymb}
\usepackage{braket}
\usepackage{makecell}
\usepackage{tablefootnote}
\usepackage{enumitem}
\usepackage[numbers,sort&compress]{natbib}

\usepackage{amsthm}
\theoremstyle{definition}
\newtheorem{theorem}{Theorem}
\newtheorem{corollary}[theorem]{Corollary}

\newtheorem{definition}[theorem]{Definition}
\newtheorem{remark}[theorem]{Remark}

\newtheorem{lemma}[theorem]{Lemma}

\usepackage{hyperref}
\hypersetup{
colorlinks=true,
citecolor=magenta}
\usepackage{cleveref}

\usepackage{tikz}
\usepackage[normalem]{ulem}
\usetikzlibrary{positioning}
\usetikzlibrary{calc}
\usetikzlibrary{decorations.pathreplacing}
\usepackage{booktabs}
\usepackage{multirow}

\newcommand{\supp}{\mathrm{supp}}
\newcommand{\Q}{\mathcal{Q}}
\newcommand{\A}{\mathcal{A}}
\newcommand{\nullity}{\mathrm{nullity}}
\newcommand{\Zlog}{\overline{Z}}
\newcommand{\Xlog}{\overline{X}}
\newcommand{\Plog}{\overline{P}}

\newcommand{\wt}{\mathrm{wt}}
\newcommand{\rs}{\mathrm{rs}}
\newcommand{\im}{\mathrm{im}}
\newcommand{\HGP}{\mathrm{HGP}}

\begin{document}

\title{Full Extractors for Logical Processing in Hypergraph Product Codes}

\author{John Blue$^{1,2}$, Zhiyang He$^3$, Hengyun Zhou$^4$, and Isaac L. Chuang$^{1,4}$}
\affiliation{$^1$Department~of~Physics,~Massachusetts~Institute~of~Technology,~Cambridge,~MA,~USA \\
$^2$The NSF AI Institute for Artificial Intelligence and Fundamental Interactions\\
$^3$Department~of~Mathematics,~Massachusetts~Institute~of~Technology,~Cambridge,~MA,~USA \\
$^4$Department~of~Electrical~Engineering~and~Computer~Science,~Massachusetts~Institute~of~Technology,~Cambridge,~MA,~USA}

\begin{abstract}
Quantum low-density parity-check (QLDPC) codes are promising candidates for practical low-overhead quantum memories. 
For large-scale fault-tolerant quantum computation, we further need logical processing methods for QLDPC codes.
In this work, we construct full extractors---surgery systems capable of measuring arbitrary logical Pauli operators on a code block---for several hypergraph product (HGP) codes.
These extractors enable logical processing via Pauli-based computation (PBC) without the compilation overhead observed in prior works.
Moreover, our extractors have sizes between 47\% and 80\% of the base HGP codes, and the extractor-augmented codes can be supported on fixed-connectivity hardware with maximum qubit degree ten.
Our approach involves assembling many partial extractors with verifiable fault tolerance into a single full extractor.
For a distance \(10\) HGP code, circuit-level noise simulations yield logical measurement error rates of approximately \(10^{-6}\) at a physical error rate of \(0.1\%\).
These results demonstrate that extractor architectures, when designed in the fixed-connectivity setting, can achieve the space efficiency of QLDPC codes without introducing compilation overhead compared to surface-code PBC architectures. 
\end{abstract}

\maketitle

\section{Introduction}

Executing large-scale quantum algorithms requires quantum error correction and fault-tolerant quantum computation (FTQC).
Extractor architectures (EAs)~\cite{he2025extractors}, based on quantum low-density parity-check (QLDPC) codes~\cite{eczoo_qldpc} and generalized code surgery~\cite{cohen2022low-overhead,cross2024improved,williamson2024low-overhead,Ide_2025,swaroop2024universal}, are a promising blueprint for practical FTQC with low space overhead. 
In an EA, the logical qubits of a quantum algorithm are partitioned into memory and compute blocks encoded in QLDPC codes. 
Each compute block is augmented with one or more \emph{extractors}, which are ancilla systems coupled to the QLDPC code to measure logical Pauli observables.
Such measurements, together with a source of high-fidelity magic states, are sufficient to implement universal FTQC via Pauli-based computation (PBC)~\cite{bravyi2016trading,litinski2019game}.
The challenge now is to move from the abstract EA framework to concrete designs tailored to specific codes and hardware platforms.

The design of extractor architectures involves tradeoffs among space overhead, compilation overhead, and connectivity.
The original extractor proposal uses full extractors with fixed bounded-degree connectivity, capable of measuring arbitrary logical Pauli operators without compilation overhead, but this generality comes with relatively large constants in the required connectivity and space overhead.

These overheads can be reduced by using partial extractors that natively measure only a restricted family of logical operators, lowering both space and connectivity requirements while introducing compilation overhead, as arbitrary measurements must be decomposed into supported ones.
Alternatively, reconfigurable connectivity, as available in atomic platforms, allows the surgery gadget to be specialized to each logical measurement, limiting the associated space and time overhead.
However, this may require designing and verifying many separate gadgets, and is less suited to fast, fixed-connectivity hardware such as superconducting qubits.
It is therefore desirable to develop practical extractor systems that achieve competitive performance along all three axes.

These tradeoffs are clearly reflected in existing constructions.
The logical processing units from the bicycle architecture~\cite{yoder2025tour,wills2026concatenating} are compact in space and can be implemented with fixed, maximum degree-seven connectivity.
Nonetheless, they support only a limited set of measurements, resulting in nontrivial compilation overheads. 
The extractors from Ref.~\cite{webster2026pinnacle} are capable of measuring all logical operators, but require reconfigurable connectivity to access different physical qubits in the code.
Most recently, Ref.~\cite{cain2026shorsalgorithm} proposed using different surgery systems for measuring different operators on a reconfigurable system. 
We provide a comparison of these systems in Table~\ref{tab:comparison}.

A further challenge in designing extractors is verifying fault tolerance.
Surgery systems are often built not with provable methods such as those in~\cite{williamson2024low-overhead,he2025extractors,yuan2026parsimoniousquantumlowdensityparitycheck}, but with overhead-minimizing heuristics, with fault tolerance verified through simulations or distance estimations.
This approach does not scale, as the number of possible measurements grows exponentially with the number of logical qubits.
Recent works have explored automated construction of surgery systems with fault-tolerance guarantees~\cite{zhou2026genecs}.
At this time, their overheads are still higher than heuristically built systems. 
This verification challenge suggests that extractors should be tailored to, and potentially co-designed with, the base code.

In this work, we follow this principle and construct full extractors on hypergraph product (HGP) codes through an assembly process tailored to the code structure.
We first construct partial extractors whose fault tolerance is easy to verify. 
We then assemble these small systems and their local fault-tolerance guarantees into a full extractor with a global fault-tolerance guarantee.
The final extractors have sizes between 47\% and 80\% of the corresponding base codes.
A whole computational block, which consists of the full extractor and the base code, can be implemented with fixed connectivity of maximum degree ten (see Table~\ref{tab:comparison}).
Notably, the methods proposed in~\cite{he2025extractors} would construct full extractors which are much bigger than the base codes.
For a distance \(10\) instance, circuit-level noise simulations of a measurement of a high-weight non-CSS logical Pauli operator with a two-stage decoder~\cite{zhao2026ultrahighratequantum} indicate logical measurement error rates of approximately \(10^{-6}\) at a physical error rate of \(0.1\%\) (see Fig.~\ref{fig:extractor-results}). 

Collectively, our results demonstrate that EAs can have practically relevant rates, can be designed within fixed, low-degree connectivity, can implement Pauli rotations without additional compilation overhead (compared to surface-code PBC architectures), and can have logical error rates comparable to base code error rates.
\begin{table}
\centering
\footnotesize
\setlength{\tabcolsep}{2pt}
\begin{tabular}{lccccccc}
\toprule
Ref. & \makecell{Code\\Family} & \(d\) & $k$ & $n_{\mathrm{tot}}$ & \(\frac{k}{n_{\mathrm{tot}}}\) & \makecell{Comp.\\Overhead} & \makecell{Max\\Degree} \\
\midrule
\multirow{4}{*}{\cite{litinski2019game}}
  & \multirow{2}{*}{\makecell{Surface\\(Fast)}}
  & 11 & 1 & 482 & \(0.21\%\) & \multirow{2}{*}{$1\times$} & \multirow{4}{*}{4} \\
  & & 17 & 1 & 1154 & \(0.09\%\) & & \\
  \cmidrule{2-7}
  & \multirow{2}{*}{\makecell{Surface\\(Compact)}} & 11 & 1 & 362 & \(0.28\%\) & \multirow{2}{*}{\(\le 8 \times\)} \\
  & & 17 & 1 & 866 & \(0.12\%\) &  \\
\midrule
\multirow{2}{*}{\cite{yoder2025tour}}
  & \multirow{2}{*}{BB}
  & 12 & 12 & 338 & \(3.55\%\) & \multirow{2}{*}{$\le 25\times$} & \multirow{2}{*}{7} \\
  & & 18 & 12 & 734 & \(1.63\%\) & & \\
\midrule
\multirow{3}{*}{\makecell{\textbf{This}\\\textbf{work}}}
  & \multirow{3}{*}{HGP}
  & 8  & 32 & 1605 & \(1.99\%\) & \multirow{3}{*}{$1\times$} & \multirow{3}{*}{10} \\
  & & 10 & 50 & 3011 &\(1.66\%\)&  & \\
  & & 16 & 50 & 5651 & \(0.88\%\) & & \\
\midrule
\multirow{3}{*}{\cite{webster2026pinnacle}} & \multirow{3}{*}{GB} & 6 & 10 & 244 & \(4.10\%\) & \multirow{3}{*}{$1\times$} & \multirow{3}{*}{Reconf.}\\
& & 10 & 12 & 452 & \(2.65\%\) & & \\
& & 16 & 14 & 860 & \(1.63\%\) & & \\
\midrule
\multirow{2}{*}{\cite{cain2026shorsalgorithm}} & BB & \(\le 18\) & 10 & 875 & \(1.14\%\) & \multirow{2}{*}{\(1\times\)} & \multirow{2}{*}{Reconf.}\\
& LP & \(\le 20\) & 148 & 3742 & \(3.96\%\) &  & \\
\bottomrule
\end{tabular}
\caption{Comparison of existing extractor and surgery systems. Code family acronyms: BB - bivariate bicycle, GB - generalized bicycle, LP - lifted product.
  The column \(n_{\mathrm{tot}}\) refers to the total number of qubits (both data qubits and check qubits) in the code block and extractor.
  ``Comp. overhead'' refers to the compilation overhead, \emph{i.e.} the maximum number of native measurements required to implement an arbitrary measurement.
  In the ``Max Degree'' column, ``Reconf.'' indicates that the construction requires reconfigurable connectivity.
  For additional details about assumptions made to obtain the values listed in the table, see Appendix~\ref{sec:table-details}.
}
\label{tab:comparison}
\end{table}

\section{Hypergraph Product Codes}

Hypergraph Product (HGP) codes~\cite{tillich2014quantum,eczoo_hypergraph_product} are a family of CSS QLDPC codes constructed from two classical codes with parity check matrices \(H_1\) and \(H_2\). We restrict our attention to the symmetric case \(H_1 = H_2 = H \in \mathbb{F}_2^{m \times n}\), for which the quantum parity check matrices become
\begin{align}
  \label{eq:hgp-pcms}
    H_X &= \left (H \otimes I_{n} \mid I_m \otimes H^\top  \right ) \\
    H_Z &= \left (I_n \otimes H \mid H^\top  \otimes I_m \right ).
\end{align}
If \(H\) and \(H^\top \) define classical codes with parameters \([n,k,d]\) and \([m, k^\top , d^\top ]\), then the quantum code \(\HGP(H, H)\) has parameters
\begin{equation}
[[n^2 + m^2, k^2 + (k^\top )^2, \min(d, d^\top )]]
\end{equation}
with the convention \(d^\top =\infty\) when \(k^\top =0\). As seen in Equation \eqref{eq:hgp-pcms}, the data qubits of an HGP code split into blocks of size \(n^2\) and \(m^2\), which are typically referred to as the ``left'' and ``right'' blocks respectively. A convenient basis of logical operators is illustrated in Fig.~\ref{fig:hgp-fig}(a). On the left block, \(Z\) (\(X\)) operators are supported on columns (rows) corresponding to an information set of \(H\).
(See Appendix~\ref{sec:cyc-codes} for the definition of an information set).
Within these columns (rows), their supports are codewords \(c \in \ker(H)\). A similar set of operators may be found on the right block, with rows and columns exchanged, and the information set and codewords now belonging to \(H^\top \). We refer to a qubit in row \(i\) and column \(j\) as an \emph{information qubit} if both \(i\) and \(j\) are in the information set for \(H\), and similarly for the right block and \(H^\top \).

Recent work~\cite{aydin2025cyclichypergraph} investigated a family of HGP codes built from binary circulant matrices as the classical parity check matrices. For such matrices, \(H\) and \(H^\top \) are permutation equivalent, and thus \(n=m\), \(k=k^\top \), and \(d=d^\top \). We use these cyclic HGP codes as the base codes for our constructions. For our logical basis, we use \(\mathcal{I}_H = \{1, \dots, k\}\) as the information set for \(H\), and \(\mathcal{I}_{H^\top } = \{\sigma(1), \dots, \sigma(k)\}\) as the information set for \(H^\top \), where \(\sigma\) is the coordinate permutation under which \(H\) and \(H^\top \) are equivalent.

\begin{figure*}
    \centering
    \includegraphics[width=\linewidth]{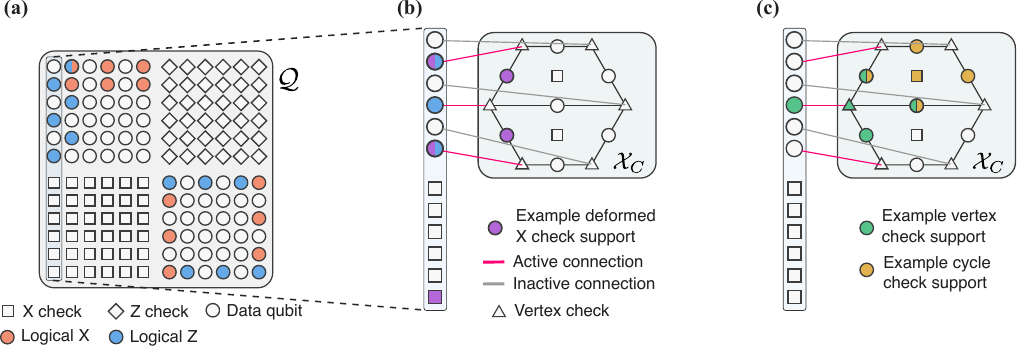}
    \caption{
    \textbf{(a)} An HGP code \(\mathcal{Q}\) with example logical \(X\) and \(Z\) operators highlighted.
    The cyclic HGP code depicted uses the check matrix \(H\) which is the binary circulant matrix corresponding to \(h(x) = 1 + x^2\) and is a \([[72,8,3]]\) code.
    \textbf{(b)} A single-column extractor \(\mathcal{X}_C\) used to measure logical \(Z\) operators supported on the first column of the left block.
	Connections activated to measure the shown logical \(Z\) operator are highlighted in pink; inactive connections are shown in light gray.
    \textbf{(c)} Example supports of a vertex check and a cycle check.
    }
    \label{fig:hgp-fig}
\end{figure*}

\section{Extractor Construction}

The extractors considered in this work for measuring a set of logical operators \(\mathcal{L}\) can be described in terms of a simple, connected graph \(G = (V, E)\), and a \emph{port function} \(f\) which assigns qubits in the support of logical operators in \(\mathcal{L}\) to vertices of \(G\). 
We call the vertices in the image of an information qubit under \(f\) information vertices.
Briefly, to measure an operator \(\overline{P} \in \mathcal{L}\) on a quantum code \(\mathcal{Q}\), the extractor is merged with \(\mathcal{Q}\) to form a new code \(\mathcal{Q}'\).
The merged code adds one qubit for every edge in \(G\), vertex checks coupling these edge qubits to the support of \(\overline{P}\), and cycle checks which correspond to elements of a cycle basis for \(G\).
Additionally, some checks of the base code must be deformed in order to ensure commutativity with the new vertex checks.
By adjusting which vertex checks are coupled to the base code, a single extractor is then capable of measuring multiple different operators.
More details about surgery can be found in Appendix~\ref{sec:code-surg} and Section~3 of~\cite{he2025extractors}.

Our construction of full extractors for cyclic HGP codes proceeds in three steps.
We begin by searching for a graph \(G = (V, E)\) with \(|V| = n\), whose corresponding extractor can measure any logical \(Z\) operator supported entirely in the first column of the left block.
We require that this single-column extractor be distance preserving, \emph{i.e.} for each logical operator supported in the first column, the merged code \(\mathcal{Q}'\) formed by attaching the extractor to the code block must have at least the distance of the original code. 
Although computing the exact distance of stabilizer codes is NP-hard~\cite{kapshikar2023hardnessminimum}, we show in Appendix~\ref{sec:single-col} that the distance of \(\mathcal{Q}'\) can be determined by looking at a much smaller code, which contains only the first column of the HGP block together with the single-column extractor.
This reduction allows the single-column extractors for the HGP codes in Table~\ref{tab:comparison} to be certified in minutes using an open-source SAT solver~\cite{ortools,cpsatlp}.

Next, we form a \(Z\)-basis extractor by connecting \(k\) single-column extractors with \(k-1\) bridge systems~\cite{cross2024improved,swaroop2024universal}.
The \(i\)th copy couples to column \(i\) of the left block via \(f\) and to the corresponding row of the right block via \(\sigma \circ f\).
Crucially, the ability to connect vertex checks to two qubits (one in the left block and one in the right), as opposed to just one as has been done in previous works, relies on the structure of the cyclic HGP code and reduces the size of the single-basis extractor by a factor of two.
This extractor is capable of measuring arbitrary logical \(Z\) operators.
The same construction, coupled instead to rows of the left block and columns of the right block, yields an \(X\)-basis extractor.
Finally, we bridge the information vertices of the \(X\)- and \(Z\)-basis extractors (Fig.~\ref{fig:extractor-construction}(b)), yielding a compact full extractor capable of measuring any logical operator on the code block, unlike previous partial extractors designed for the fixed-connectivity setting~\cite{yoder2025tour}.

\begin{figure}
    \centering
    \includegraphics{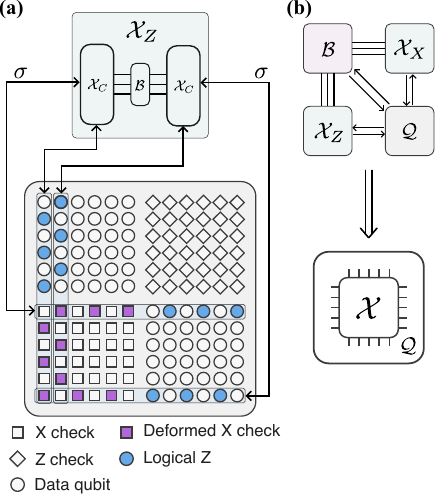}
    \caption{
    \textbf{(a)} A \(Z\) extractor (\(\mathcal{X}_Z\)), capable of measuring any logical \(Z\) operator on the base code \(\mathcal{Q}\).
    The \(Z\) extractor is formed by connecting \(k\) single-column extractors with bridges.
	Lines between the extractor and base code show how they are connected.
    The connections between \(\mathcal{X}_Z\) and the rows of the right block are governed by the permutation \(\sigma\).
    \textbf{(b)} Connecting an \(X\) extractor (\(\mathcal{X}_X\)) and a \(Z\) extractor (\(\mathcal{X}_Z\)) with a bridge \(\mathcal{B}\) to form a full extractor.
    }
    \label{fig:extractor-construction}
\end{figure}

\begin{theorem}
    \label{thm:main}
    Let \(\mathcal{Q}\) be a \([[2n^2, 2k^2, d]]\) cyclic HGP code. Suppose \(\Q\) has a distance-preserving single-column extractor with associated graph \(G = (V, E)\), where \(|V| = n\). Then, there exists a full extractor for \(\Q\), which uses \(\Theta(k(|E| + |V| + k))\) additional qubits and checks. Using this extractor, there exists a protocol to measure any logical operator of \(\Q\) with phenomenological fault distance at least
    \begin{equation}
        \min\left(d, k^2, k\cdot \lambda(G)\right)
    \end{equation}
    where, given the \(|V| \times |E|\) incidence matrix \(M\) of \(G\), \(\lambda(G)\) is the edge connectivity of \(G\), defined as
    \begin{equation}
        \lambda(G) = \min \{ |M^\top v| : v \in \mathbb{F}_2^{|V|},\ |v| \notin \{0, |V|\}\}.
    \end{equation}
\end{theorem}

Theorem~\ref{thm:main} reduces the design of a full extractor for cyclic HGP codes to the search for a distance-preserving single-column extractor, whose graph \(G\) has \(\lambda(G) \gtrsim d/k\). The edge connectivity of a graph is efficiently computable~\cite{beineke2013topicsstructuralgraph}, and the distance preserving property can be verified straightforwardly in practice using the reduction described above. A full definition of phenomenological fault distance can be found in Appendix~\ref{sec:meas}, and a proof of the theorem, along with a detailed accounting of the connectivity requirements of the construction, can be found in Appendix~\ref{sec:extractor-construction}. 

\section{Circuit-Level Noise Simulations}
\label{sec:simulation-results}

To evaluate our constructions, we perform circuit-level noise simulations of the extractor for the \([[882,50,10]]\) code.
In particular, we simulate the measurement of a randomly sampled logical operator denoted \(\overline{P}\), with logical weight \(37\) and physical weight \(197\).
We use a standard circuit-level depolarizing noise model in which single-qubit, two-qubit, initialization, and measurement errors occur with probability \(p\), and idle errors occur with probability \(p/10\). To decode, we use a two-stage decoder similar to the one from \cite{zhao2026ultrahighratequantum}.
In the first stage, we use the Relay-BP decoder~\cite{muller2025improved}.
If the Relay-BP decoder does not converge, we fall back to an integer programming decoder \cite{landahl2011faulttolerant,cain2024correlateddecoding}.

Using Stim~\cite{gidney2021stim}, we first noiselessly prepare the base code in a joint \(+1\) eigenstate of \(50\) independent logical operators, including \(\overline{P}\).
We then noisily measure the checks of the merged code for \(N_R\) rounds, noisily measure the base code syndromes for one round, and finish by noiselessly measuring the base code checks and the other \(49\) logical operators.

First, keeping \(N_R=15\) fixed, we sweep over physical error rates, and compare against a \(Z\)-basis memory experiment with the base code using \(16\) rounds of syndrome measurement (Fig.~\ref{fig:extractor-results}(a)).
At each tested physical error rate, the block error rate for the extractor is within a factor of \(2.2\) of the base code memory error rate, demonstrating that the extractor does not significantly change the protection of the information encoded in the other \(49\) logical qubits in the block.
Additionally, at a physical error rate of \(p=0.1\%\), the logical measurement error rate is \(8 \times 10^{-7}\) \((95\%\,\mathrm{CI}: 0.4\text{--}1.6\times 10^{-6})\).

Next, keeping the physical error rate fixed at \(p=0.1\%\), we sweep over \(N_R\), finding as expected that the logical measurement error rate decreases as \(N_R\) increases, and the block error rate increases with \(N_R\) (Fig.~\ref{fig:extractor-results}(b)).
The crossing point appears to occur at \(N_R = 15\) at a logical error rate of roughly \(10^{-6}\), suggesting an architecture based on this extractor could reach the megaquop regime~\cite{preskill2025megaquop}.

\begin{figure*}
    \centering
    \includegraphics[width=1.0\linewidth]{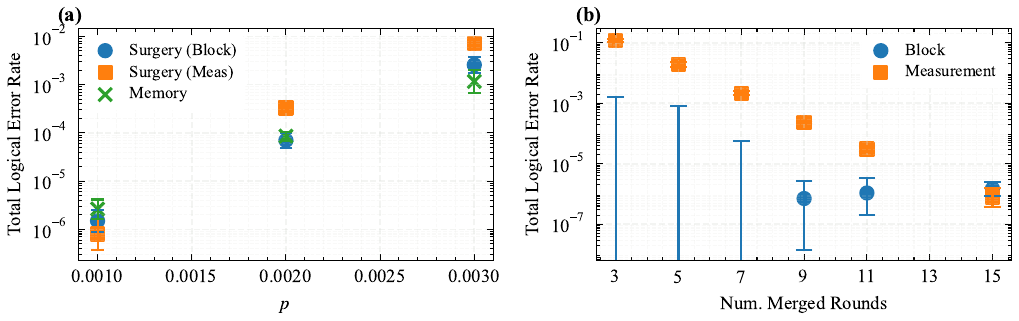}
    \caption{Logical error rates of the \([[882, 50, 10]]\) HGP code and extractor system used to measure a randomly sampled element of \(\{I, X, Y, Z\}^{\otimes 50}\). 
    For the extractor system, ``Meas'' refers to the error rate of the logical measurement, and ``Block'' refers to the logical error rate of the other \(49\) logical qubits during the experiment.
    Note that the shown logical error rates are not normalized per round or per logical qubit.
    Error bars indicate \(95\%\) Agresti-Coull confidence intervals.
    \textbf{(a)} A comparison of the logical error rates for a sixteen-round memory experiment with the base HGP code, and a surgery experiment with fifteen rounds of noisy syndrome measurement in the merged code.
    Note that the surgery system has error rates similar to memory.
    \textbf{(b)} A comparison of the block and measurement error rate of the surgery system with a varying number of rounds of noisy syndrome measurement, with a physical error rate of \(p=0.1\%\). There were no observed block errors when the number of rounds of merged code syndrome measurement was equal to three, five, and seven rounds.
    Both experiments used a two-stage decoder described in the main text.
    }
    \label{fig:extractor-results}
\end{figure*}

\section{Discussion and Conclusion}
In this work, we have shown that extractors for finite-sized HGP codes can simultaneously support arbitrary logical measurements and bounded connectivity, with a total footprint less than \(2\times\) that of the base code.
While we have only focused on extractors for a single compute block, using the techniques described in~\cite{he2025extractors}, it is straightforward to connect many such blocks together with bridges, and to magic state factories with adaptors, to yield a full architecture.
These results suggest that fixed-connectivity platforms can take advantage of the space savings provided by QLDPC codes, without incurring a time-overhead penalty relative to surface-code PBC architectures.

Realizing these gains with superconducting qubits, however, requires connectivity beyond the planar layouts of current systems.
Recent progress in fabrication capabilities shows that the pieces required for such a system are becoming available.
In particular, developments in 3D integration~\cite{rosenberg20173dintegrated,yost2020solidstatequbits,kosen2022buildingblocks} and long-range couplers~\cite{marxer2023longdistance,xiong2025scalablelow,xu2026tunablehybrid} suggest a multi-tier design in which the connectivity graph is spread across multiple layers to reduce congestion~\cite{tremblay2022constant-overhead}. Recent work has begun to investigate how to lay out QLDPC codes on such a design~\cite{mathews2026placingandrouting}.

We note that while we have focused on fixed-connectivity implementation throughout this paper, our extractors can be implemented on reconfigurable hardware, and they overcome the verification challenge mentioned in the introduction.

Our work opens several immediate questions.
First, the practical decoding of such systems remains a challenge.
While the two-stage decoder used in this work is sufficient for benchmarking purposes, a decoder relying on an ILP solver is unlikely to meet the \(\sim 1\,\mu\mathrm{s}\) syndrome measurement time demonstrated by superconducting hardware~\cite{acharya2024quantum}. 
Recent work has demonstrated that tuning the Relay-BP parameters specifically for surgery circuits can result in significant improvement in logical error rates~\cite{wills2026concatenating}; it would be interesting to see how much further tuning could close the gap with the two-stage decoder.

Furthermore, while HGP codes have constant encoding rates, there are several code families with better parameters than the cyclic HGP codes studied here, \emph{e.g.} various families of two-block group algebra codes~\cite{kovalev2013quantumkronecker,panteleev2021degenerate,bravyi2024high-threshold2,lin2024quantumtwoblock}, lifted~\cite{panteleev2022quantum} and balanced~\cite{breuckmann2021balanced} product codes, and recently introduced ultra-high-rate constructions~\cite{kasai2026breakingorthogonality,zhao2026ultrahighratequantum}.
Extending our techniques of assembling an extractor from smaller components to such codes could yield fixed-connectivity architectures with lower space overheads, potentially at smaller block sizes.
It would also be interesting to incorporate recently developed parallel surgery techniques~\cite{zhang2025time-efficient,cowtan2025parallel,zheng2025high,gu2026qgpu} and fast surgery techniques~\cite{cowtan2025fast,baspin2025fast,chang2026constant} into the design of extractors. 

Finally, our constructions enable resource estimations in the megaquop regime with HGP-based extractor architectures. 
The full measurement capability of our extractors provides a solid layer of abstraction for compiling algorithms onto such an EA.
Compilation onto EAs is an open direction with many fresh problems and opportunities~\cite{he2025extractors,yoder2025tour,wills2026concatenating,xu2026distilling,liu2026assessing,sethi2026injeqt,sethi2026optimizing,khan2026architecting}.

\section*{Acknowledgements}
We thank Adam Wills, Qian Xu, Robert Huang, Ted Yoder, AJ Davenport, and Andy Zeyi Liu for insightful discussions and comments. 
J.B. acknowledges support from the MIT Center for Quantum Engineering--Laboratory for Physical Sciences Doc Bedard Fellowship.
Z.H. acknowledges support from the MIT Department of Mathematics, the MIT-IBM Watson AI Lab, and the NSF
Graduate Research Fellowship Program under Grant No. 2141064.
This work is supported in part by the National Science Foundation under Cooperative Agreement PHY-2019786 (The NSF AI Institute for Artificial Intelligence and Fundamental Interactions, \url{https://iaifi.org}) and the DARPA HARQ program, and made use of resources provided by subMIT at MIT Physics~\cite{bendavid2025submitphysicsanalysisfacility}, as well as the FASRC Cannon cluster supported by the FAS Division of Science Research Computing Group at Harvard University.
The authors also acknowledge the MIT Office of Research Computing and Data for providing high performance computing resources that have contributed to the research results reported within this paper. This research was, in part, funded by the U.S. Government. The views and conclusions contained in this document are those of the authors and should not be interpreted as representing the official policies, either expressed or implied, of the U.S. Government.

\section*{Author Contribution and AI Use Statements}
ZH developed the initial idea of constructing full extractors for HGP codes. JB and ZH developed the extractor design and proofs of fault-tolerance, with additional input from HZ and IC. JB wrote the code with input from all other authors and ran the simulations. All authors contributed to scientific discussions and the writing of the manuscript.

We used large language models (LLMs) to assist with developing the simulation code, preparing figures, and editing the manuscript.
LLMs also identified errors in earlier versions of the proofs of Theorem~\ref{thm:single-basis-dist} and Theorem~\ref{thm:full-ext}; the corrected proof of Theorem~\ref{thm:full-ext} incorporates a fix proposed by an LLM. The authors have verified all results and take full responsibility for the content of the paper.

\bibliographystyle{quantum}
\bibliography{main}

\onecolumn
\appendix

\section*{Appendices}
In the Appendices, we start by defining notation (Appendix~\ref{sec:notation}), providing the necessary definitions from graph theory (Appendix~\ref{sec:graph-theory}) and classical coding theory (Appendix~\ref{sec:cyc-codes}), and summarizing the relevant results regarding hypergraph product codes (Appendix~\ref{sec:hgp-codes}) and extractors (Appendix~\ref{sec:code-surg}).
We then provide the details of our extractor construction and a proof of its fault tolerance (Appendix~\ref{sec:extractor-construction}).
Next, we provide details on our numerical simulations and some additional results (Appendix~\ref{sec:numerics-details}).
Finally, we detail the assumptions made to obtain the values listed in Table~\ref{tab:comparison} (Appendix~\ref{sec:table-details}).

\section{Notation}
\label{sec:notation}
First, we establish notation that we will use throughout the rest of the appendices.
\begin{itemize}
\item \([n]\) refers to the set \(\{1, 2, \dots, n\}\).
\item \(e_j\) refers to the unit vector with a one in the \(j\)th entry and zeros everywhere else.
\item \(E_{i,j}\) refers to a matrix with a one in the \(i\)th row and \(j\)th column, and zeros everywhere else.
\item For a matrix \(F\), we use \(F_{(i,\cdot)}\) to denote the \(i\)th row, and \(F_{(\cdot, j)}\) to denote the \(j\)th column.
\item For a vector \(v\) or matrix \(V\), \(|v|\) and \(|V|\) denote the Hamming weight, \emph{i.e.} the number of non-zero entries.
\item We will often use the same letter to denote a vector over \(\mathbb{F}_2\), and the set of its nonzero indices. This allows us to write statements such as \(|u + v| = |u| + |v| - 2 |u \cap v|\) where on the left hand side we think of \(u\) and \(v\) as vectors, and on the right hand side as sets.
\item Given a set of Pauli operators \(A\), we denote the weight of the minimum weight element of \(A\) as \(d(A) = \min \{ \wt(M) : M \in A\}\). We let \(d_X(A)\) denote the minimum weight Pauli operator in \(A\) which consists only of tensor products of \(I\) and \(X\), and similarly for \(d_Z(A)\).
\item Given a function \(f: A \to B\) with \(C \subset A\), we will often abbreviate \(f|_{C}\) as \(f_C\).
\end{itemize}

We adopt the convention that indices for qubits, vectors, and matrices begin with one. In the discussion of cyclic codes, we will index polynomial coordinates starting with zero.

\section{Graph Theory}
\label{sec:graph-theory}
In this section we provide several definitions from graph theory that we will use throughout the appendices. Let \(G = (V, E)\) be a connected, simple graph. Label the vertices \(v_1, \dots, v_{|V|}\) and the edges \(e_1, \dots, e_{|E|}\). The \emph{incidence matrix} of \(G\) is a matrix \(M \in \mathbb{F}_2^{|V| \times |E|}\) such that
\begin{equation}
    M_{ij} = \begin{cases}
        1 & \text{ if } v_i \in e_j \\
        0 & \text{ otherwise}
    \end{cases}
\end{equation}
The \emph{cycle space} of \(G\) is defined as the kernel of \(M\), and a basis for the cycle space is called a \emph{cycle basis}. The cycle space of a connected graph has dimension \(|E| - |V| + 1\). Let \(N\) be a matrix whose rows form a cycle basis for \(G\). The maximum row weight of \(N\) is called the maximum weight of the cycle basis, and the maximum column weight of \(N\) is called the \emph{congestion} of the cycle basis.

Let \(S \subset V\). A \emph{path matching} of \(S\) is a set of edges such that in the edge-induced subgraph, every vertex in \(S\) has odd degree, and every vertex in \(V \setminus S\) has even degree.

Finally, the \emph{edge connectivity} \(\lambda(G)\) of a graph is given by
\begin{equation}
    \lambda(G) = \min \{ |M^\top v| : v \in \mathbb{F}_2^{|V|},\ |v| \notin \{0, |V|\}\}.
\end{equation}

\section{Classical Cyclic Codes}
\label{sec:cyc-codes}
We briefly recall several relevant facts about cyclic codes --- a more thorough introduction can be found in an introductory text on error correction (\emph{e.g.} \cite{huffman2003fundamentals}).
A classical linear code \(C \subset \mathbb{F}_2^n\) is said to be cyclic if, for any codeword \(c = (c_0, c_1, \dots, c_{n-1}) \in C\), the cyclic shift \(c_{\rightarrow 1} = (c_{n-1}, c_0, \dots, c_{n-2}) \in C\) as well.
For each vector \(v \in \mathbb{F}_2^n\) one can associate a polynomial \(v(x) = v_0 +  v_1x + \dots +  v_{n-1}x^{n-1} \in \mathbb{F}_2[x]\).
The requirement that \(C\) is closed under cyclic shifts of codewords can be shown to be equivalent to \(C\) (when viewed as a set of polynomials) forming an ideal of \(\mathbb{F}_2[x] / (x^n-1)\).
For every non-zero cyclic code, there exists a generator polynomial \(g(x)\) and a check polynomial \(h(x)\) such that
\begin{itemize}
    \item \(C = \langle g(x) \rangle \)
    \item \(g(x) h(x) = x^n-1\)
    \item \(k = n - \mathrm{deg}(g)\)
    \item The generator matrix \(G\) is given by 
    \begin{equation}
        G = \begin{pmatrix}
            g_0 & g_1 &  &\dots & g_{n-k} & 0 & \dots & 0 & 0 \\
            0 & g_0 & & \dots & g_{n-k-1} & g_{n-k} &\dots & 0 & 0 \\
            \vdots & \vdots & & \ddots & \vdots & \vdots & \ddots  & \vdots & \vdots  \\
            0 & 0 & & \dots& & & & g_{n-k-1} & g_{n-k}
        \end{pmatrix}
    \end{equation}
    \item The parity check matrix \(H\) is given by
    \begin{equation}
    \label{eq:cyclic-pmat}
    H = \begin{pmatrix}
            h_k & h_{k-1} &  &\dots & h_0 & 0 & \dots & 0 & 0 \\
            0 & h_k & & \dots & h_1 & h_0 &\dots & 0 & 0 \\
            \vdots & \vdots & & \ddots & \vdots & \vdots & \ddots  & \vdots & \vdots  \\
            0 & 0 & & \dots& & & & h_1 & h_0
        \end{pmatrix}
    \end{equation}
\end{itemize}

For any classical code with generator matrix \(G\), a set of \(k\) coordinates whose corresponding columns in \(G\) are linearly independent is called an \emph{information set}.
For a cyclic code, any \(k\) consecutive coordinates form an information set.

\section{Hypergraph Product Codes}
\label{sec:hgp-codes}
\subsection{Construction}
Let \(\mathcal{C}_A\) and \(\mathcal{C}_B\) be \([n_A, k_A, d_A]\) and \([n_B, k_B, d_B]\) classical codes with parity check matrices \(\partial_A \in \mathbb{F}_2^{m_A \times n_A}\) and \(\partial_B \in \mathbb{F}_2^{m_B \times n_B}\). For each code, we have an associated one-term chain complex \(C_{\bullet, x}\)
\begin{equation}
  C_{x}^1 \xrightarrow{\partial_{x}} C_{x}^0
\end{equation}
for \(x \in \{A, B\}\). Here, \(C_{x}^1 = \mathbb{F}_2^{n_x}\) and \(C_{x}^0 = \mathbb{F}_2^{m_x}\).

For each code, we can also define the transpose code \(\mathcal{C}_x^\top \) as the classical code with the parity check matrix \(\partial_x^\top \). We denote the number of bits encoded in \(\mathcal{C}_x^\top \) as \(k_x^\top \), and its distance as \(d_x^\top \). Note that \(\mathcal{C}_x^\top \) is not uniquely defined by \(\mathcal{C}_x\), and instead depends on the specific parity check matrix \(\partial_x\). The transpose code is associated with the one-term cochain complex \(C^{\bullet}_x\):
\begin{equation}
  C_x^1 \xleftarrow{\partial_x^\top } C_x^0.
\end{equation}

Consider the two-term chain complex \(Q_{\bullet}\) obtained by taking the product of \(C_{\bullet,A}\) and \(C^{\bullet}_{ B}\):
\begingroup\renewcommand{\theHequation}{Qbullet}
\begin{equation}
\label{eq:q-cmplx}
\begin{tikzpicture}
  \node (t2) {\(C^1_A \otimes C^0_B\)};
  \node (t4) [right=2.5 of t2] {\(Q_2\)};
  \coordinate[left=1.0 of t2] (t1);
  \coordinate[right=1.0 of t2] (t3);
  \node (m1) [below=1.5 of t1] {\(C^1_A \otimes C^1_B\)};
  \node (m3) [below=1.5 of t3] {\(C^0_A \otimes C^0_B\)};
  \node (b2) [below=3.0 of t2] {\(C^0_A \otimes C^1_B\)};
  \node (m4) at (t4|-m3) {\(Q_1\)};
  \node (b4) at (t4|-b2) {\(Q_0\)};
  \draw[->] (t2) -- (m1) node[midway, above left] {\small \(I_{n_A} \otimes \partial_B^\top \)};
  \draw[->] (t2) -- (m3) node[midway, above right] {\small \(\partial_A \otimes I_{m_B}\)};
  \draw[->] (m1) -- (b2) node[midway, below left] {\small \(\partial_A \otimes I_{n_B}\)};
  \draw[->] (m3) -- (b2) node[midway, below right] {\small \(I_{m_A} \otimes \partial_B^\top \)};
  \draw[->] (t4) -- (m4) node[midway, right] {\small \(H_Z^\top \)};
  \draw[->] (m4) -- (b4) node[midway, right] {\small \(H_X\)};
\end{tikzpicture}
\tag{\(Q_{\bullet}\)}
\end{equation}
\endgroup

The hypergraph product (HGP) code~\cite{tillich2014quantum} \(\mathrm{HGP}(\partial_A, \partial_B)\) is defined as the CSS code \(\mathcal{Q}\) obtained from the chain complex \(Q_2 \xrightarrow{H_Z^\top } Q_1 \xrightarrow{H_X} Q_0\). The parity check matrices are given by
\begin{equation}
\label{eq:hgp-par}
H_X = ( \partial_A \otimes I_{n_B} | I_{m_A} \otimes \partial_B^\top ), \qquad H_Z = (I_{n_A} \otimes \partial_B | \partial_A^\top  \otimes I_{m_B}).
\end{equation}

We call the block of \(n_A \cdot n_B\) qubits associated with \(C_A^1 \otimes C_B^1\) the \emph{left block}, and the block of \(m_A\cdot m_B\) qubits associated with \(C_A^0 \otimes C_B^0\) the \emph{right block}. At different points, we will find it useful to represent sets of data qubits in the left (right) block as 
\begin{itemize}
\item Subsets of \([n_A \cdot n_B]\) (\([m_A \cdot m_B]\))
\item Vectors in \(C_A^1 \otimes C_B^1\) (\(C_A^0 \otimes C_B^0\))
\item Matrices in \(\mathbb{F}_2^{n_A \times n_B}\) (\(\mathbb{F}_2^{m_A \times m_B}\)).
\end{itemize}
In the grid picture shown in Figure \ref{fig:hgp-logop}, the left qubit in row \(i\) and column \(j\) corresponds to the matrix \(E_{i,j}\), the vector \(e_i \otimes e_j\), and the element \(n_B \cdot (i - 1) + j \in [n_A \cdot n_B]\). 
In the reverse direction, each data qubit \(l \in [n_A \cdot n_B]\) in the left block corresponds to a unit vector \(e_i \otimes e_j \in C_A^1 \otimes C_B^1\) and matrix \(E_{i,j}\) where
\begin{equation}
  \label{eq:index-mapping}
  i = \left \lceil \frac{l}{n_B}\right \rceil, \quad j = l - n_B(i-1).
\end{equation}

Additionally, under this mapping each subset of data qubits \(s \subset [n_A \cdot n_B]\) in the left block corresponds to a vector \(v \in C_A^1 \otimes C_B^1\) such that \(\supp(v) = s\) (where we are implicitly using the mapping defined in Equation \eqref{eq:index-mapping}). Similarly, each data qubit in the right block \(r \in [m_A \cdot m_B]\) corresponds to a unit vector \(e_i \otimes e_j \in C_A^0 \otimes C_B^0\) where now \(i = \lceil \frac{r}{m_B}\rceil\) and \(j = r - m_B(i - 1)\), and each subset of data qubits in the right block \(t \subset [m_A m_B]\) corresponds to a vector \(u \in C_A^0 \otimes C_B^0\) such that \(\supp(u) = t\). Note that we will often abuse notation and use the same symbol to refer to vectors over \(\mathbb{F}_2\) and their support.

To denote Pauli operators acting on the qubits, we adopt the notation from~\cite{quintavalle2022reshapedecoder} that for matrices \(L \in \mathbb{F}_2^{n_A \times n_B}\) and \(R \in \mathbb{F}_2^{m_A \times m_B}\),
\begin{equation}
  \label{eq:x-op-notation}
  X(L, R) = \left ( \bigotimes_{i_a, i_b} X^{L_{i_a, i_b}} \right) \otimes \left ( \bigotimes_{j_a, j_b} X^{R_{j_a, j_b}} \right )
\end{equation}
that is, \(L\) and \(R\) give the support of the \(X\) operator on the left and right block. We adopt similar notation for \(Z\) operators.

To analyze the logical operators of \(\mathcal{Q}\), we must analyze the homology and cohomology groups of the chain complex. The K\"{u}nneth formula tells us that the first homology group of \(\mathcal{Q}\), and thus the space of logical \(Z\) operators, is given by
\begin{equation}
\label{eq:log-z}
(H_{A, 1} \otimes H_B^0) \oplus (H_{A, 0} \otimes H_B^1)
\end{equation}
where \(H_{A, i}\) (\(H_A^i\)) denotes the \(i\)th homology (co-homology) group of the chain complex associated with the classical code \(\mathcal{C}_A\) (and similarly for \(\mathcal{C}_B\)).
Similarly, the first cohomology group, and space of logical \(X\) operators, is given by
\begin{equation}
\label{eq:log-x}
(H_A^0 \otimes H_{B, 1}) \oplus (H_A^1 \otimes H_{B, 0}).
\end{equation}

A straightforward application of the rank-nullity theorem reveals that \(\dim(H_{x,1}) = \dim(H_x^0) = k_x\), and \(\dim(H_{x,0}) = \dim(H_x^1) = k_x^\top \). Thus, \(\mathcal{Q}\) encodes \(k_Ak_B + k_A^\top  k_B^\top \) logical qubits. Further analysis of the elements of the first homology and cohomology groups (see, \textit{e.g.} Appendix C of~\cite{campbell2019theory}) shows that their Hamming weights are lower bounded by \(\min(d_A, d_B, d_A^\top , d_B^\top )\) (where we use the convention that a code has \(d=\infty\) if \(k=0\)).

\subsection{Cyclic HGP Codes}
\label{sec:cyclic-hgp}

All of the new constructions from~\ref{tab:comparison} use the HGP of classical cyclic codes, first explored in \cite{aydin2025cyclichypergraph}, as the base code. More specifically, given a check polynomial \(h(x)\), we construct \(\Q = \mathrm{HGP}(\partial, \partial)\) where we take \(\partial\) to be the parity check matrix whose rows consist of all cyclic shifts of \(h^*(x)\)
\footnote{Given a degree \(a\) polynomial \(p(x) = p_0 + x p_1 + \dots + x^{a}p_{a}\), the reciprocal polynomial \(p^*(x)\) is obtained by reversing the coefficients, \emph{i.e.} \(p^*(x) = p_{a} + x p_{a-1} + \dots + x^{a} p_0\)},
as opposed to just the first \(n-k\) as shown in Equation \eqref{eq:cyclic-pmat}.
Note that taking the transpose of \(\partial\) results in a matrix whose rows are generated by cyclic shifts of \(h(x)\), and so \(\partial^\top \) defines the cyclic code with check polynomial \(h^*(x)\). 
This code is permutation equivalent to \(C\) (the code defined by \(\partial\) as a parity check matrix). 
Thus, assuming that \(C\) is an \([n, k, d]\) code, \(C^\top  = \ker{\partial^\top }\) will also be an \([n,k,d]\) code, and \(\mathrm{HGP}(\partial, \partial)\) will be a \([[2n^2, 2k^2, d]]\) code. We list several examples of such cyclic HGP codes in Table \ref{tab:cyclic-hgp}.
\begin{table}[h!]
{
\setlength{\tabcolsep}{8pt}
\begin{center}
    \begin{tabular}{cccc}
    \toprule
     \(h(x)\)  & \([n_c, k_c, d_c]\) & \([[n, k, d]]\) & Rate \\
         \midrule
       \(1 + x^3 + x^4\)       & \([15, 4, 8]\)  & \([[450, 32, 8]]\)   & \(7.1\%\) \\ 
       \(1 + x + x^5\)       & \([21, 5, 10]\) & \([[882, 50, 10]]\)  & \(5.7\%\) \\
       \(1 + x^4 + x^5 + x^6\) & \([31, 6, 15]\) & \([[1922, 72, 15]]\) & \(3.7\%\) \\
       \(1 + x^2 + x^5\)       & \([31, 5, 16]\) & \([[1922, 50, 16]]\) & \(2.6\%\) \\
         \bottomrule
    \end{tabular}
    \caption{A selection of cyclic HGP codes. The header \([n_c, k_c, d_c]\) refers to the parameters of the classical cyclic code, and the header \([[n, k, d]]\) refers to the parameters of the quantum code. Note that the first two codes are from \cite{aydin2025cyclichypergraph}, and the first and last codes are members of the simplex family~\cite{eczoo_simplex}. The \([[1922,50,16]]\) HGP code was studied in~\cite{panteleev2021degenerate}.}
    \label{tab:cyclic-hgp}
\end{center}
}
\end{table}

\subsection{Canonical Basis of Logical Operators}
\label{sec:hgp-logical-ops}
Following~\cite{quintavalle2023partitioning} and~\cite{xu2024fast}, we now construct a well-structured symplectic basis of logical operators for \(\mathrm{HGP}(\partial_A, \partial_B)\). We begin by explicitly writing out the homology and cohomology groups in terms of \(\partial_x\):
\begin{align}
\label{eq:3}
   H_{x, 1} &= \ker(\partial_x) & H_{x, 0} &= \mathbb{F}_2^{m_x} / \im(\partial_x) \\
   H_x^1 &= \ker \left(\partial_x^\top \right) & H_x^0 &= \mathbb{F}_2^{n_x} / \im \left(\partial_x^\top \right).
\end{align}

By equations (\ref{eq:log-z}) and (\ref{eq:log-x}), to find a basis of logical operators it suffices to independently find bases for \(H_{x, 0}, H_{x,1}, H_x^0\) and \(H_x^1\). 
We begin by considering the generator matrix \(G_x\) for \(\mathcal{C}_x\).
Assume we order the coordinates of \(\mathcal{C}_x\) such that the first \(k\) columns of \(G_x\) are linearly independent (\emph{i.e.} they form an information set).
Then by performing row reductions we can write \(G_x\) as
\begin{equation}
\label{eq:2}
G_x = \left (\begin{array}{c|c} I_{k_x} & A_x^\top  \end{array}\right)
\end{equation}
with corresponding parity check matrix
\begin{equation}
\label{eq:1}
\partial_x' = \left (\begin{array}{c|c} A_x & I_{n_x-k_x} \end{array}  \right ).
\end{equation}
Note that \(\ker(\partial_x) = \ker(\partial_x')\).

Let \(\left\{ c_{x,i} \right\}_{i=1}^{k_x}\) be the rows of \(G_x\), which form a basis for \(\ker(\partial_x) = H_{x,1}\).
Let \(e_{x,j} \in \mathbb{F}_2^{n_x}\) be the unit vectors in \(\mathbb{F}_2^{n_x}\) with \((e_{x, j})_k = \delta_{jk}\). We can see that \(\left\{ e_{x,j} \right\}_{j=1}^{k_x}\) forms a basis for \(\left( \im (\partial_x^\top ) \right)^{\bullet}\), and thus \(H_x^0\)~\footnote{The notation \(W^{\bullet}\) denotes a \emph{subspace complement}: For a vector space \(V\) and subspace \(W\), the subspace complement \(W^{\bullet}\) is another subspace \(T\) such that \(W \cap T = \left\{ 0 \right\}\) and \(W \oplus T = V\). If \(\left\{ t_i \right\}\) forms a basis for \(W^{\bullet}\), then \(t_i + W\) forms a basis for \(V/W\).}.

Furthermore, by construction, \(c_{x,i}^\top  e_{x,j} = \delta_{i,j}\). By similarly writing a full rank row reduction of \(\partial_x^\top \) in systematic form, we can obtain transpose codewords \(\{\tilde{c}_{x,i}\}_{i=1}^{k_x^\top }\) and unit vectors \(\left\{ e_{x,j} \right\}_{j=1}^{k_x^\top } \subset \mathbb{F}_2^{m_x}\) that form bases for \(H_x^1\) and \(H_{x,0}\), and where \(\tilde{c}_{x,i}^\top  e_{x,j} = \delta_{i,j}\).

We can then consider four sets of logical operators (illustrated in \Cref{fig:hgp-logop}):
\begin{itemize}
\item \(\mathcal{L}_Z^L\): Logical \(Z\) operators with support only on the left block. These are spanned by \[\left\{ \bar{Z}^L_{ij} = Z(c_{A,i} \cdot e^\top _{B,j}, 0): i \in [k_A], j \in [k_B] \right\}.\]
    \item \(\mathcal{L}_X^{L}\): Logical \(X\) operators with support only on the left block. These are spanned by \[\left\{ \bar{X}^L_{ij} = X(e_{A,i} \cdot c^\top _{B,j}, 0): i \in [k_A], j \in [k_B]  \right\}.\]
    \item \(\mathcal{L}_Z^{R}\): Logical \(Z\) operators with support only on the right block. These are spanned by \[\left\{ \bar{Z}^R_{ij} =Z(0, e_{A,i} \cdot \tilde{c}^\top _{B,j}): i \in \left [k_A^\top \right], j \in \left[ k_B^\top  \right ] \right\}.\]
    \item \(\mathcal{L}_X^R\): Logical \(X\) operators with support only on the right block. These are spanned by \[\left\{ \bar{X}^R_{ij} = X(0, \tilde{c}_{A,i} \cdot e^\top _{B,j}) : i \in \left [k_A^\top \right], j \in \left [k_B^\top  \right ]\right\}.\]
\end{itemize}

%
%
%
%
%
%

\begin{figure}[htbp]
\centering
\begin{tikzpicture}[
  font=\footnotesize,
  scale=1.0,
  every node/.style={inner sep=1pt},
]

\def\csz{0.46}   
\def\LN{6}       
\def\RN{4}       
\def\Roff{3.3}   

%
%
%

\fill[RoyalBlue!12] (0,0) rectangle ({2*\csz},{-6*\csz});
\fill[BrickRed!12]  (0,0) rectangle ({6*\csz},{-2*\csz});

\foreach \i in {1,3,5}{
  \fill[RoyalBlue!60] (0,{-(\i-1)*\csz}) rectangle (\csz,{-\i*\csz});
}
\foreach \j in {1,2,5}{
  \fill[BrickRed!55] ({(\j-1)*\csz},0) rectangle ({\j*\csz},{-\csz});
}
\fill[Plum!65] (0,0) rectangle (\csz,{-\csz});

\foreach \i in {0,1,2,3,4,5,6}{
  \draw[gray!45, very thin] (0,{-\i*\csz}) -- ({6*\csz},{-\i*\csz});
}
\foreach \j in {0,1,2,3,4,5,6}{
  \draw[gray!45, very thin] ({\j*\csz},0) -- ({\j*\csz},{-6*\csz});
}
\draw[black, thick] (0,0) rectangle ({6*\csz},{-6*\csz});

\node[left=3pt] at (0,{-0.5*\csz}) {$1$};
\node[left=3pt] at (0,{-1.5*\csz}) {$2$};
\node[left=3pt] at (0,{-2.5*\csz}) {$3$};
\node[left=3pt, text=gray] at (0,{-3.7*\csz}) {\raisebox{4pt}{$\vdots$}};
\node[left=3pt] at (0,{-5.5*\csz}) {\tiny $n_{\!A}$};

\node[above=2pt] at ({0.5*\csz},0) {$1$};
\node[above=2pt] at ({1.5*\csz},0) {$2$};
\node[above=2pt] at ({2.5*\csz},0) {$3$};
\node[above=2pt, text=gray] at ({3.7*\csz},0) {$\cdots$};
\node[above=2pt] at ({5.5*\csz},0) {\tiny $n_{\!B}$};

\node[font=\small\bfseries, anchor=base] at (1.38, 1.40) {Left block\vphantom{g}};

\draw[decorate, decoration={brace,amplitude=3pt}, RoyalBlue!70, thick]
  (0, 0.42) -- (0.92, 0.42)
  node[midway, above=3pt, RoyalBlue!80, font=\scriptsize]
  {info cols ($k_B$)};

\draw[decorate, decoration={brace,mirror,amplitude=3pt}, BrickRed!70, thick]
  (-0.45, 0) -- (-0.45, -0.92)
  node[midway, left=4pt, BrickRed!80, font=\scriptsize, align=right]
  {info\\rows ($k_A$)};

\draw[->, RoyalBlue!80, thick]
  (-1.5, -1.38) -- (-0.07, -1.38);
\node[left=2pt, RoyalBlue!90, align=right, font=\scriptsize] at (-1.55, -1.38)
  {$\bar{Z}^L_{11}, c_{A,1}\!\in\!\ker\partial_A$};

\node[left=2pt, BrickRed!90, align=right, font=\scriptsize] (xLnode) at (-1.55, -2.5)
  {$\bar{X}^L_{11}, c_{B,1}\!\in\!\ker\partial_B$};
\draw[->, BrickRed!80, thick] (-1.5, -2.5) to [out=20, in=230] (2.07, -0.23);

\draw[thin, ->] (-1.5, 1.05) to [out=270, in=135] (0.18, -0.18);
\node[font=\scriptsize, anchor=base east] at (-0.75, 1.15) {info qubit $(i,j)$};

%
%

\fill[RoyalBlue!12] (\Roff,0) rectangle ({\Roff+4*\csz},{-2*\csz});
\fill[BrickRed!12]  (\Roff,0) rectangle ({\Roff+2*\csz},{-4*\csz});

\foreach \j in {1,3,4}{
  \fill[RoyalBlue!60] ({\Roff+(\j-1)*\csz},0) rectangle ({\Roff+\j*\csz},{-\csz});
}
\foreach \i in {1,3}{
  \fill[BrickRed!55] (\Roff,{-(\i-1)*\csz}) rectangle ({\Roff+\csz},{-\i*\csz});
}
\fill[Plum!65] (\Roff,0) rectangle ({\Roff+\csz},{-\csz});

\foreach \i in {0,1,2,3,4}{
  \draw[gray!45, very thin] (\Roff,{-\i*\csz}) -- ({\Roff+4*\csz},{-\i*\csz});
}
\foreach \j in {0,1,2,3,4}{
  \draw[gray!45, very thin] ({\Roff+\j*\csz},0) -- ({\Roff+\j*\csz},{-4*\csz});
}
\draw[black, thick] (\Roff,0) rectangle ({\Roff+4*\csz},{-4*\csz});

\node[left=3pt] at (\Roff,{-0.5*\csz}) {$1$};
\node[left=3pt] at (\Roff,{-1.5*\csz}) {$2$};
\node[left=3pt, text=gray] at (\Roff,{-2.7*\csz}) {\raisebox{8pt}{$\vdots$}};
\node[left=0pt] at (\Roff,{-3.5*\csz}) {\tiny $m_{\!A}$};

\node[above=2pt] at ({\Roff+0.5*\csz},0) {$1$};
\node[above=2pt] at ({\Roff+1.5*\csz},0) {$2$};
\node[above=2pt, text=gray] at ({\Roff+2.7*\csz},0) {\hspace*{-6pt}$\cdots$};
\node[above=2pt] at ({\Roff+3.5*\csz},0) {\tiny $m_{\!B}$};

\node[font=\small\bfseries, anchor=base] at (4.22, 1.40) {Right block\vphantom{g}};

\draw[decorate, decoration={brace,amplitude=3pt}, RoyalBlue!70, thick]
  (5.24, 0) -- (5.24, -0.92)
  node[midway, right=5pt, RoyalBlue!80, font=\scriptsize]
  {info rows ($k_A^{\top}$)};

\draw[decorate, decoration={brace,amplitude=3pt}, BrickRed!70, thick]
  (3.3, 0.42) -- (4.22, 0.42)
  node[midway, above=3pt, xshift=-16pt, BrickRed!80, font=\scriptsize]
  {info cols ($k_B^{\top}$)};

\draw[->, RoyalBlue!80, thick]
  (5.70, 0.75) to [out=270, in=45] (4.91, -0.23);
\node[RoyalBlue!90, align=left, font=\scriptsize, anchor=south west]
  at (4.80, 0.80)
  {$\bar{Z}^R_{11}$,\; $\tilde{c}_{B,1}\!\in\!\ker\partial_B^{\top}$};

\draw[->, BrickRed!80, thick]
  (6.54, -1.15) -- (3.83, -1.15);
\node[right=2pt, BrickRed!90, align=left, font=\scriptsize]
  at (6.59, -1.15)
  {$\bar{X}^R_{11}$\\[-1pt]$\tilde{c}_{A,1}\!\in\!\ker\partial_A^{\top}$};

\def\LegY{-7.2*\csz}
\def\LegX{0.15}

\fill[RoyalBlue!60] (\LegX,\LegY) rectangle ({\LegX+\csz},{\LegY-\csz});
\node[right=3pt, font=\scriptsize] at ({\LegX+\csz},{\LegY-0.5*\csz})
  {$\bar{Z}$ support (codeword of $\ker\partial$ or $\ker\partial^{\top}$)};

\fill[BrickRed!55] (\LegX,{\LegY-1.5*\csz}) rectangle ({\LegX+\csz},{\LegY-2.5*\csz});
\node[right=3pt, font=\scriptsize] at ({\LegX+\csz},{\LegY-2.0*\csz})
  {$\bar{X}$ support (codeword of $\ker\partial$ or $\ker\partial^{\top}$)};

\fill[Plum!65] (\LegX,{\LegY-3.0*\csz}) rectangle ({\LegX+\csz},{\LegY-4.0*\csz});
\node[right=3pt, font=\scriptsize] at ({\LegX+\csz},{\LegY-3.5*\csz})
  {information qubit (unique $\bar{Z}$--$\bar{X}$ overlap)};

\end{tikzpicture}
\caption{
Canonical basis of logical operators for $\mathrm{HGP}(\partial_A,\partial_B)$,
  illustrated on a small representative instance ($n_A{=}n_B{=}6$, $m_A{=}m_B{=}4$).
  Each logical qubit $(i,j)$ is associated with an \emph{information qubit} at which its
  paired $\bar{Z}$ and $\bar{X}$ operators intersect.
  On the left block, $\bar{Z}^L_{ij}$ occupies information column~$j$ (sparse, along the
  codeword $c_{A,i}\!\in\!\ker\partial_A$) while $\bar{X}^L_{ij}$ occupies information
  row~$i$ (sparse, along $c_{B,j}\!\in\!\ker\partial_B$).
  On the right block the roles of rows and columns are exchanged, with codewords drawn from
  the transpose codes $\ker\partial_A^{\top}$ and $\ker\partial_B^{\top}$.
  Light shading indicates the information rows/columns; the unique $\bar{Z}$--$\bar{X}$
  intersection (violet) is the information qubit.}
\label{fig:hgp-logop}
\end{figure}
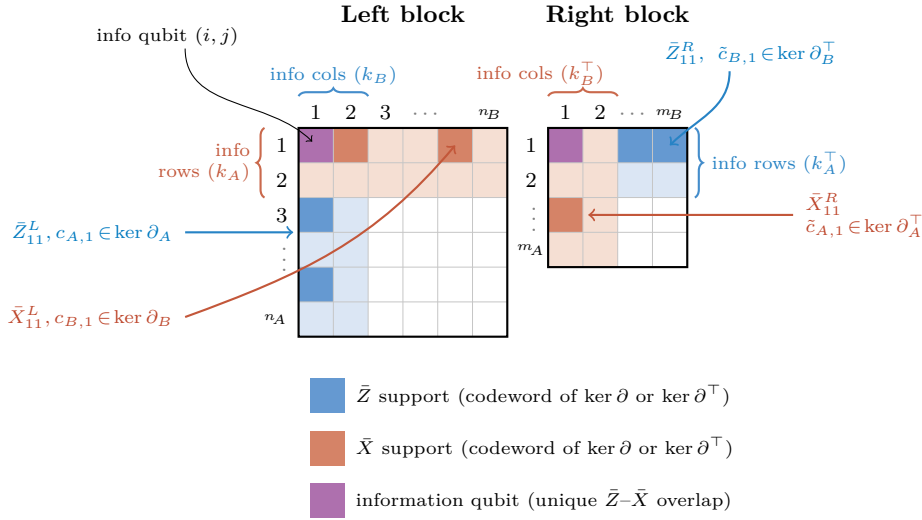

To verify that this basis has the correct commutation relations, consider the operators \(\bar{Z}^L_{i,j}\) and \(\bar{X}^L_{k,l}\). We can see that these operators overlap if and only if \(c_{A,i}^\top  e_{A,k} = 1\) and \(e_{B,j}^\top  c_{B,l} = 1\), which is only the case when \(i=k\) and \(j=l\). A similar argument shows the same holds for the logical operators on the right block. Not only does this mean that this basis has the proper commutation relations, it also shows that we can identify each logical qubit in the code, defined by a pair of logical operators, with the single physical qubit on which these operators have overlapping support. We refer to this qubit as an \emph{information} qubit, as it corresponds to the product of two information bits of the classical codes in the hypergraph product.

We conclude this discussion with a definition and two lemmas that will be useful in our proof of fault tolerance. 

\begin{definition}
  For a Pauli operator \(Z(L, R)\), we define
  \begin{equation}
    \label{eq:rowlogZ}
    \mathrm{row}_{\mathrm{log}}(L) = \{ i \in [n_A] : L_{(i, \cdot)} \notin \im(\partial_B^\top ) \}
  \end{equation}
  and
  \begin{equation}
    \label{eq:collogZ}
    \mathrm{col}_{\mathrm{log}}(R) = \{ j \in [m_B] : R_{(\cdot, j)} \notin \im(\partial_A) \}.
  \end{equation}
Similarly, for a Pauli operator \(X(L, R)\), we define
\begin{equation}
    \label{eq:collogX}
    \mathrm{col}_{\mathrm{log}}(L) = \{ j \in [n_B] : L_{(\cdot, j)} \notin \im(\partial_A^\top ) \}
  \end{equation}
  and
  \begin{equation}
    \label{eq:rowlogX}
    \mathrm{row}_{\mathrm{log}}(R) = \{ i \in [m_A] : R_{(i, \cdot)} \notin \im(\partial_B) \}.
  \end{equation}

\end{definition}

The following was first shown in the proof of Proposition 3 of \cite{quintavalle2022reshapedecoder}.
\begin{lemma}
	\label{lem:z-hom-inv}
	Let \(\overline{Z} = Z(L, R)\) be a logical operator of \(\mathrm{HGP}(\partial_A, \partial_B)\). Then \(\mathrm{row}_{\mathrm{log}}(L)\) and \(\mathrm{col}_{\mathrm{log}}(R)\) are invariants of the homology class of \(\overline{Z}\). Similarly, given a logical operator \(\overline{X} = X(L, R)\), \(\mathrm{col}_{\mathrm{log}}(L)\) and \(\mathrm{row}_{\mathrm{log}}(R)\) are invariants of the homology class of \(\overline{X}\).
\end{lemma}

\begin{proof}
	We only prove the result for \(Z(L, R)\) as the reasoning for \(X(L, R)\) is almost identical. Let \(u \in C_A^1 \otimes C_B^0\) correspond to a \(Z\) check \(M\). We can then see that the logical operator \(\overline{Z}M = Z(L + U \partial_B, R + \partial_A U)\). Each row of \(U\partial_B\) lies in the image of \(\partial_B^\top \), thus if row \(i\) of \(L\) is in \(\mathrm{Im}(\partial_B^\top )\), then so is row \(i\) of \(L + U \partial_B\). However, if row \(i\) of \(L\) is not in \(\mathrm{Im}(\partial_B^\top )\) (\emph{i.e.} it is equal to \(t + v\) for \(v \in (\mathrm{Im}(\partial_B^\top ))\), \(t \in \left ( \mathrm{Im}(\partial_B^\top )\right )^{\bullet}, t \ne 0\)), then neither will row \(i\) of \(L + U \partial_B\). Thus, \(\mathrm{row}_{\mathrm{log}}(L) = \mathrm{row}_{\mathrm{log}}(L + U\partial_B)\). Similar reasoning shows that \(\mathrm{col}_{\mathrm{log}}(R) = \mathrm{col}_{\mathrm{log}}(R + \partial_A U)\).
\end{proof}

Next, we have the following result about cleaning logical operators of HGP codes:
\begin{lemma}
  \label{lem:hgp-clean}
  Let \(\Zlog = Z(L, 0)\) for \(L = \sum_{j=1}^{k_B} c_j \cdot e_j^\top \), where each \(c_j \in \ker(\partial_A)\) and let \(v \subset [k_B]\) denote some set of columns of the left block of qubits. Suppose there exists a \(j' \notin v\) such that \(c_{j'} \ne 0\). Then every representative of \(\Zlog\) will have support of size at least \(d\) outside of the columns corresponding to \(v\).
\end{lemma}

\begin{proof}
  Let \(M = Z(U\partial_B, \partial_AU)\) be the product of \(Z\) checks. We can see that \(\Zlog M = Z(L + U\partial_B, \partial_A U)\). Let \(L' = L + U\partial_B\), and consider the \(i\)th row of \(L'\), where \(i \in c_{j'}\). Since \(L_{i, j'} = 1\), the only way that \(L'_{i, j'} = 0\) is if \((U\partial_B)_{(i, j')} = 1\). Note that \((U \partial_B)_{(i, \cdot)} \in \rs(\partial_B)\). By considering the standard form of \(\partial_B\), we can see that since \(j' \in [k_B]\), if \((U\partial_B)_{(i, j')} = 1\), then there exists some \(j'' > k_B\) such that \((U\partial_B)_{(i, j'')} = 1\). Thus, the \(i\)th row of \(L'\) has at least one nonzero entry outside of the columns of \(v\). Since \(c_{j'} \in \ker(\partial_A)\) and \(c_{j'} \ne 0\), we have that \(\Zlog M\) has support of size at least \(d\) outside of the columns of \(v\), as desired.
\end{proof}

Note that we can make similar statements about operators of the form \(\Zlog = Z(0, R), \Xlog = X(L, 0)\), and \(\Xlog = X(0, R)\). Put together, we have the following:

\begin{corollary}
\label{cor:clean}
  Let \(\Plog\) and \(\Plog'\) be products of canonical basis operators for \(\mathrm{HGP}(\partial_A, \partial_B)\), such that \(\Plog'\) has support outside of the rows and columns supporting \(\Plog\). Then every representative of \(\Plog'\) has support at least \(d\) outside of the rows and columns supporting \(\Plog\). 
\end{corollary}

\section{Code Surgery}
\label{sec:code-surg}

We now briefly overview the construction of ancilla systems that can be used to measure logical operators of LDPC codes.
We also refer readers to Section~3 of~\cite{he2025extractors} for a more detailed review of recent works and methods.

Let $\Q$ be an $[[n,k,d]]$ QLDPC code with stabilizer $\mathcal{S}$, and let $Q$ be the set of qubits in $\Q$.
Let $\Plog$ be a logical operator of $\Q$.
Given \(q \in \supp(\Plog)\), let \(\Plog(q)\) be the single qubit Pauli operator acting on \(q\).
Our goal is to construct a new QLDPC code $\Q'$ with stabilizer $\mathcal{S}'$ such that $\Plog \in \mathcal{S}'$.
Thus, by deforming the code $\Q$ into $\Q'$, we can measure $\Plog$.

To determine the new code $\Q'$, we will make use of ancilla systems, as defined below.
\begin{definition}[Ancilla System]
	An ancilla system for a QLDPC code $\Q$ and logical operator $\Plog$ consists of the following:
	\begin{itemize}
		\item An $[[n_{\A}, k_{\A}, d_{\A}]]$ QLDPC code $\mathcal{A} = \mathrm{CSS}(H_X, H_Z)$, such that $\dim \left ( \ker \left ( H_Z^\top  \right ) \right ) = 1$.
			Denote the set of $X$ checks of $\mathcal{A}$ as $A_X$ and the set of $Z$ checks as $A_Z$.
            We associate a hypergraph $G = (V, E)$ to the parity check matrix \(H_Z^\top \), where the rows/checks of \(H_Z\) are associated with the vertices, and the columns/qubits are associated with edges.
		\item A port function $f: \supp(\Plog) \to A_Z$ such that for $0 \ne c \in \ker(H_Z^\top )$ we have $\mathrm{im}(f) \subset \supp(c)$.
        \item A deformation map \(D: \mathcal{S}\rightarrow \mathcal{P}(E)\) (where \(\mathcal{P}\) denotes the power set). For each $S \in \mathcal{S}$, and \(a \in A_Z\), let \(K(S, a) = \{ q \in f^{-1}(\{a\}): \{S(q), \Plog(q)\} = 0\}\), and then let \(K(S) = \{ a \in A_Z : |K(S, a)| \text{ is odd}\}\). $D(S)$ is a path matching for $K(S)$ in $G$.
        Note that if \(K(S) = \emptyset\), \(D(S) = \emptyset\).
	\end{itemize}
\end{definition}

\begin{remark}
    We note that previous works have required that the port function be injective. Here, we lift that requirement and will consider non-injective port functions. 
\end{remark}

\begin{remark}
  Previous works~\cite{zhang2025time-efficient,cowtan2025parallel,zheng2025high,gu2026qgpu} have considered the case when \(\mathrm{dim}(\ker(H_Z^\top )) > 1\), thus allowing a single ancilla system to measure multiple operators in parallel. In this work, we will restrict our attention to the single operator case, leaving the construction of ``high rate'' extractors to the future.
\end{remark}

\begin{remark}
  In all of the ancilla systems considered in this work, we will ensure that \(\rs(H_X) = \ker(H_Z)\), and so \(k_{\A} = 0\).
\end{remark}

Using the ancilla system, we can now define the code $\Q'$ for which $\Plog$ is an element of the stabilizer.

\begin{definition}[Deformed Code]
	The deformed code of a QLDPC code $\Q$ and ancilla system $\mathcal{A}$ is a quantum code $\Q '$ consisting of the qubits of $\Q$ and $\A$, along with the checks generated by $\mathcal{S}'$ defined as follows:
	\begin{itemize}
		\item For each $S \in A_Z$, add $ \prod_{q \in f^{-1}(\{S\})} \Plog(q)\cdot S$ to $\mathcal{S}'$. These are called \textbf{vertex checks}.
		\item For each $S \in A_X$, add $S$ to $\mathcal{S}'$. These are called \textbf{cycle checks} or \textbf{gauge checks}.
        \item For each \(S \in \mathcal{S}\), add \(S \prod_{e \in D(S)} X(e)\) to \(\mathcal{S}'\). These are called \textbf{path matching checks}.
	\end{itemize}
\end{definition}

All ancilla systems in this paper can be described by a graph, in which case it is guaranteed that the necessary path matchings exist.
To see that \(\Plog \in \mathcal{S}'\), consider the product of the vertex checks associated with the support of \(c \in \ker(H_Z^\top )\).
Since \(\im(f) \subset \supp(c)\), by construction this product has support equal to \(\Plog\) on the base code, and no support on the ancilla system.

\begin{definition}[Distance-preserving]
   We say that an ancilla system \((\A, f, D)\) is \emph{distance preserving} when the distance of the deformed code \(\Q'\) is at least as large as the distance of \(\Q\). 
\end{definition}

As described above, a single ancilla system is capable of measuring a single logical operator, which naively leads to the requirement that one design \(4^k-1\) ancilla systems in order to be able to measure the entire Pauli group via surgery. 
This challenge motivated the proposal of \textit{extractor} systems, which we present next.

\begin{definition}[Extractor]
    An extractor for a QLDPC code \(\Q\) and set of logical operators \(\{\Plog_i\}\) consists of
    \begin{itemize}
        \item An additional \([[n_{\mathcal{X}}, k_{\mathcal{X}}, d_{\mathcal{X}}]]\) QLDPC code \(\mathcal{X} = \mathrm{CSS}(H_X, H_Z)\) such that \(\mathrm{dim}(\ker(H_Z^\top )) = 1\). Denote the \(Z\) checks of \(\mathcal{X}\) as \(X_Z\).
        \item A function \(f: \bigcup_i \supp(\Plog_i) \to X_Z\)
        \item A deformation map \(D: \mathcal{S}\rightarrow \mathcal{P}(E)\) such that for each \(\Plog_i\) and \(S \in \mathcal{S}\), \(D(S)\) contains a path matching of \(K(S)\), where \(K(S)\) is defined using the port function \(f|_{\supp(\Plog_i)}\).
        Let \(D_{\Plog_i}: \mathcal{S} \to \mathcal{P}(E)\) which takes \(S\) to such a path matching.
    \end{itemize}
    such that \((\mathcal{X}, f|_{\supp(\Plog_i)}, D_{\Plog_i})\) is a valid ancilla system for \(\mathcal{Q}\) and \(\Plog_i\). We call an extractor \emph{distance preserving} when each \((\mathcal{X}, f|_{\supp(\Plog_i)}, D_{\Plog_i})\) is distance preserving. We further specify that an extractor is a \emph{full} extractor when the set of logical operators \(\{\Plog_i\}\) is equal to \emph{all} of the logical Pauli operators of \(\mathcal{Q}\), and a \emph{partial} extractor otherwise.
\end{definition}

\Cref{fig:code-surgery} summarizes the definitions of this section and their relationships.

%
%
%

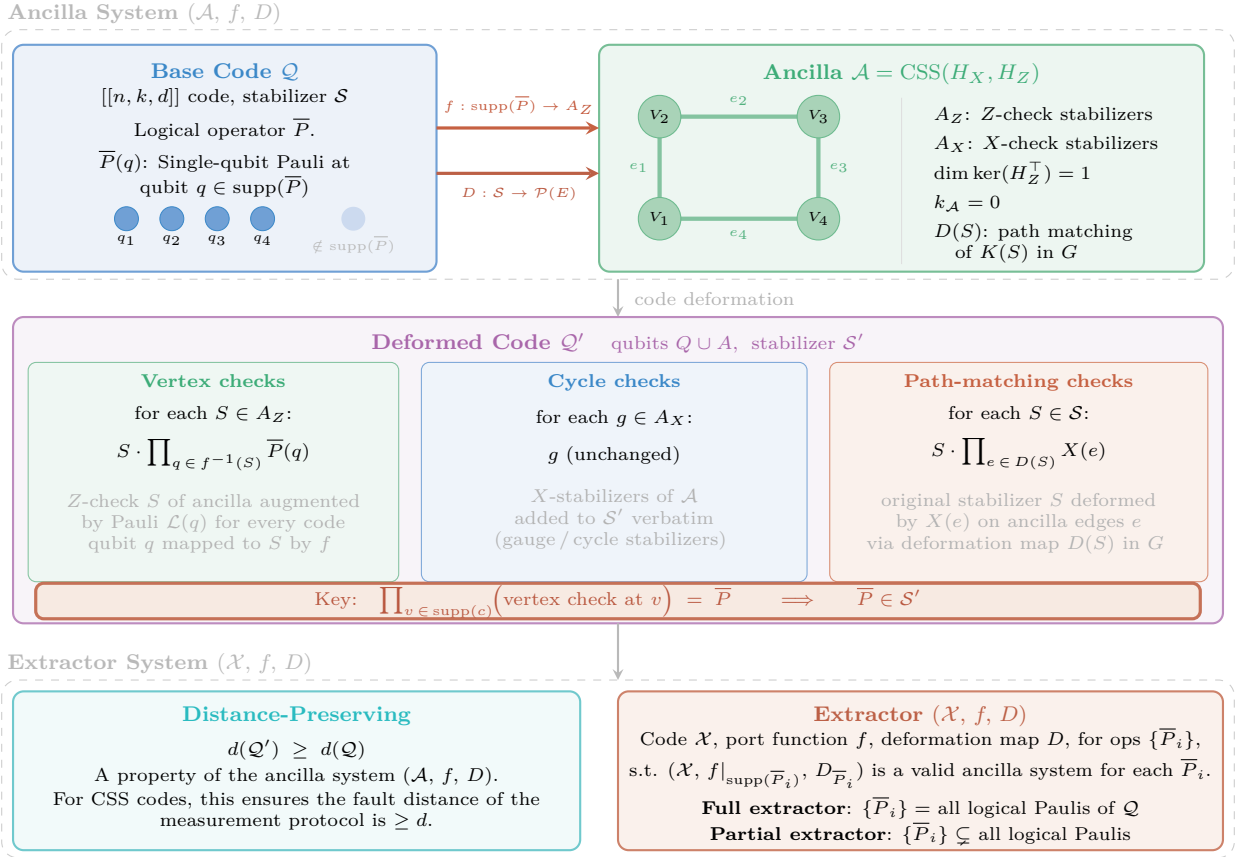
\begin{figure*}[t]
\centering
\begin{tikzpicture}[
  font=\scriptsize,
  >=stealth,
  every node/.style={inner sep=2pt},
]


\draw[draw=gray!40, rounded corners=6pt, dashed]
  (-0.15, 0.20) rectangle (16.15, -3.1);
\node[font=\footnotesize\bfseries, gray!55, anchor=south west]
  at (-0.15, 0.20) {Ancilla System $(\mathcal{A},\,f,\,D)$};

\draw[draw=RoyalBlue!55, fill=RoyalBlue!6, rounded corners=4pt, thick]
  (0, 0) rectangle (5.6, -3.0);
\node[font=\footnotesize\bfseries, RoyalBlue!75, anchor=north]
  at (2.8, -0.17) {Base Code $\mathcal{Q}$};
\node[align=center] at (2.8, -1.3) {
  $[[n,k,d]]$ code, stabilizer \(\mathcal{S}\)\\[3pt]
  Logical operator $\Plog$.\\[3pt]
  $\Plog(q)$: Single-qubit Pauli at \\qubit $q \in \supp(\Plog)$
};
\foreach \x in {1.5, 2.1, 2.7, 3.3, 4.5}{
  \fill[RoyalBlue!18] (\x, -2.3) circle (4.5pt);
}
\foreach \x in {1.5, 2.1, 2.7, 3.3}{
  \fill[RoyalBlue!55] (\x, -2.3) circle (4.5pt);
  \draw[RoyalBlue!75, very thin] (\x, -2.3) circle (4.5pt);
}
\foreach \x/\q in {1.5/$q_1$, 2.1/$q_2$, 2.7/$q_3$, 3.3/$q_4$}{
  \node[font=\tiny, below=4pt] at (\x, -2.3) {\q};
}
\node[font=\tiny, gray!55, below=4pt] at (4.5, -2.3) {$\notin \supp(\Plog)$};

\draw[draw=ForestGreen!55, fill=ForestGreen!6, rounded corners=4pt, thick]
  (7.75, 0) rectangle (15.750, -3.0);
\node[font=\footnotesize\bfseries, ForestGreen!75, anchor=north]
  at (11.75, -0.17) {Ancilla $\mathcal{A} = \mathrm{CSS}(H_X, H_Z)$};

\draw[->, BrickRed!75, very thick] (5.6, -1.1) -- (7.75, -1.1);
\node[above=3pt, BrickRed!80] at (6.69, -1.1)
  {\tiny $f : \supp(\Plog) \to A_Z$};
\draw[->, Mahogany!75, very thick] (5.6, -1.7) -- (7.75, -1.7);
\node[below=3pt, Mahogany!80] at (6.69, -1.7)
  {\tiny $D : \mathcal{S} \to \mathcal{P}(E)$};

\draw[gray!25, thin] (11.75, -0.7) -- (11.75, -2.9);

\draw[ForestGreen!50, line width=1.6pt]
  (8.55,-2.30) -- (8.55,-0.95);   
\draw[ForestGreen!50, line width=1.6pt]
  (8.55,-0.95) -- (10.65,-0.95);   
\draw[ForestGreen!50, line width=1.6pt]
  (10.65,-0.95) -- (10.65,-2.30);   
\draw[ForestGreen!50, line width=1.6pt]
  (10.65,-2.30) -- (8.55,-2.30);   
\node[font=\tiny, ForestGreen!70, left=2pt]  at (8.55,  -1.625) {$e_1$};
\node[font=\tiny, ForestGreen!70, above=2pt] at (9.6, -0.95)  {$e_2$};
\node[font=\tiny, ForestGreen!70, right=2pt] at (10.65,  -1.625) {$e_3$};
\node[font=\tiny, ForestGreen!70, below=2pt] at (9.6, -2.30)  {$e_4$};
\foreach \vx/\vy/\vl in
  {8.55/-2.30/V_1, 8.55/-0.95/V_2, 10.65/-0.95/V_3, 10.65/-2.30/V_4}{
  \fill[ForestGreen!35] (\vx,\vy) circle (8pt);
  \draw[ForestGreen!60, thin] (\vx,\vy) circle (8pt);
  \node[font=\tiny\bfseries] at (\vx,\vy) {$\vl$};
}


\node[align=left, anchor=north west] at (12.1, -0.75) {
  $A_Z$: $Z$-check stabilizers 
\\[3pt]
  $A_X$: $X$-check stabilizers 
\\[3pt]
  $\dim \ker(H_Z^{\top}) = 1$ 
\\[3pt]
  $k_{\mathcal{A}} = 0$ 
\\[3pt]
  $D(S)$: path matching\\
  \hspace*{1em}of $K(S)$ in $G$
};

\draw[->, thick, gray!55] (8.0, -3.1) -- (8.0, -3.6);
\node[right=4pt, gray!60] at (8.0, -3.35) {code deformation};

\draw[draw=Plum!50, fill=Plum!4, rounded corners=4pt, thick]
  (0, -3.6) rectangle (16.0, -7.65);
\node[font=\footnotesize\bfseries, Plum!70, anchor=north] at (8.0, -3.73)
  {Deformed Code $\mathcal{Q}'$\quad
   \normalfont\scriptsize qubits $Q \cup A$,\; stabilizer $\mathcal{S}'$};

\draw[draw=ForestGreen!50, fill=ForestGreen!5, rounded corners=3pt]
  (0.2, -4.2) rectangle (5.1, -7.1);
\node[font=\scriptsize\bfseries, ForestGreen!75, anchor=north]
  at (2.65, -4.28) {Vertex checks};
\node[align=center] at (2.65, -5.75) {
  for each $S \in A_Z$:\\[6pt]
  $S \cdot \prod_{q\,\in\,f^{-1}(S)} \Plog(q)$\\[8pt]
  \color{gray!60}$Z$-check $S$ of ancilla augmented\\
  \color{gray!60}by Pauli $\mathcal{L}(q)$ for every code\\
  \color{gray!60}qubit $q$ mapped to $S$ by $f$
};

\draw[draw=RoyalBlue!50, fill=RoyalBlue!5, rounded corners=3pt]
  (5.4, -4.2) rectangle (10.5, -7.1);
\node[font=\scriptsize\bfseries, RoyalBlue!75, anchor=north]
  at (7.95, -4.28) {Cycle checks};
\node[align=center] at (7.95, -5.75) {
  for each $g \in A_X$:\\[6pt]
  $g$ \textrm{(unchanged)}\\[8pt]
  \color{gray!60}$X$-stabilizers of $\mathcal{A}$\\
  \color{gray!60}added to $\mathcal{S}'$ verbatim\\
  \color{gray!60}(gauge\,/\,cycle stabilizers)
};

\draw[draw=BrickRed!50, fill=BrickRed!5, rounded corners=3pt]
  (10.8, -4.2) rectangle (15.8, -7.1);
\node[font=\scriptsize\bfseries, BrickRed!75, anchor=north]
  at (13.3, -4.28) {Path-matching checks};
\node[align=center] at (13.3, -5.75) {
  for each $S \in \mathcal{S}$:\\[6pt]
  $S \cdot \prod_{e\,\in\,D(S)} X(e)$\\[8pt]
  \color{gray!60}original stabilizer $S$ deformed\\
  \color{gray!60}by $X(e)$ on ancilla edges $e$\\
  \color{gray!60}via deformation map $D(S)$ in $G$
};

\draw[draw=Mahogany!65, fill=Mahogany!8, rounded corners=3pt, very thick]
  (0.3, -7.12) rectangle (15.7, -7.59);
\node[font=\scriptsize, Mahogany!80] at (8.0, -7.36) {
  Key:\quad
  $\prod_{v\,\in\,\mathrm{supp}(c)}\!\bigl(\text{vertex check at }v\bigr)
  \;=\;\Plog$
  \qquad$\Longrightarrow$\qquad
  $\Plog\in\mathcal{S}'$\quad
};

\draw[->, thick, gray!55] (8.0, -7.65) -- (8.0, -8.40);


\draw[draw=gray!40, rounded corners=6pt, dashed]
  (-0.15, -8.40) rectangle (16.15, -10.75);
\node[font=\footnotesize\bfseries, gray!55, anchor=south west]
  at (-0.15, -8.40) {Extractor System $(\mathcal{X},\,f,\,D)$};

\draw[draw=TealBlue!60, fill=TealBlue!5, rounded corners=4pt, thick]
  (0, -8.55) rectangle (7.5, -10.65);
\node[font=\footnotesize\bfseries, TealBlue!80] at (3.75, -8.87)
  {Distance-Preserving};
\node[align=center] at (3.75, -9.9) {
  $d(\mathcal{Q}') \;\geq\; d(\mathcal{Q})$\\[2pt]
  A property of the ancilla system $(\mathcal{A},\,f,\,D)$.\\[1pt]
  For CSS codes, this ensures the fault distance of the\\
  measurement protocol is $\geq d$.\\
};

\draw[draw=Mahogany!55, fill=Mahogany!5, rounded corners=4pt, thick]
  (8.0, -8.55) rectangle (16.0, -10.65);
\node[font=\footnotesize\bfseries, Mahogany!75] at (12.0, -8.87)
  {Extractor $(\mathcal{X},\,f,\,D)$};
\node[align=center] at (12.0, -9.8) {
  Code $\mathcal{X}$, port function $f$, deformation map $D$, for ops $\{\Plog_i\}$,\\[2pt] 
   s.t. $(\mathcal{X},\,f|_{\supp(\Plog_i)},\,D_{\Plog_i})$ is a valid ancilla system for each $\Plog_i$.\\[2pt]
  \textbf{Full extractor}: $\{\Plog_i\} =$ all logical Paulis of $\mathcal{Q}$\\
  \textbf{Partial extractor}: $\{\Plog_i\} \subsetneq$ all logical Paulis
};

\end{tikzpicture}
\caption{
Structure of code surgery and its constituent definitions
  (Appendix~\ref{sec:code-surg}).
  An \emph{ancilla system} $(\mathcal{A},f,D)$ consists of a CSS ancilla code $\mathcal{A}$
  (drawn as a graph: vertices $= A_Z$, edges $=$ ancilla qubits, cycles $= A_X$)
  together with a \emph{port function} $f:\supp(\Plog)\to A_Z$ satisfying
  $\mathrm{im}(f)\subset\mathrm{supp}(c)$, where $c$ is the unique nonzero element of
  $\ker H_Z^{\top}$, and a \emph{deformation map} $D:\mathcal{S}\to\mathcal{P}(E)$ assigning
  each base-code stabilizer a path matching of $K(S)$ in the ancilla graph $G$.
  Logical measurement is performed by measuring the stabilizers of the \emph{deformed code} $\mathcal{Q}'$
  whose checks can be classified into three groups: \emph{vertex checks} (each $Z$-check of $\mathcal{A}$
  augmented by the Paulis $\mathcal{L}(q)$ for code qubits mapped to it by $f$);
  \emph{cycle checks} ($X$-checks of $\mathcal{A}$, unchanged); and
  \emph{path-matching checks} (original stabilizers $\mathcal{S}$, deformed onto 
  edges $D(S)$ in the ancilla graph).  The product of all vertex checks
  indexed by $\mathrm{supp}(c)$ equals $\Plog$, making $\Plog$ an element of the stabilizer of
  $\mathcal{Q}'$.
  An \emph{extractor} generalizes an ancilla system so that one code
  $\mathcal{X}$ handles many logical operators simultaneously, ideally in a way which preserves distance.}
\label{fig:code-surgery}
\end{figure*}

\subsection{Measurement Procedure and Fault-Distance}
\label{sec:meas}

Let \(\Q\) be a code of distance \(d\) with logical operator \(\overline{P}_M\), and ancilla system \(\A\) with \(k_{\A} = 0\). In order to use \(\A\) to measure \(\overline{P}_M\), we use the following code-deformation protocol:
\begin{enumerate}
\item Measure the stabilizers of \(\Q\) for \(d\) rounds.
\item Initialize the edge qubits of the ancilla system in \(\ket{+}\).
\item Measure the stabilizers of the merged code \(\Q'\) for \(d\) rounds.
\item Measure the edge qubits of the ancilla system in the \(X\) basis.
\item Measure the stabilizers of \(\Q\) for \(d\) rounds.
\end{enumerate}
We take the product of the measurement results of the vertex checks corresponding to \(\supp(c)\), for \(0 \ne c \in \ker(H_Z^\top )\) from the first round of measurement of step three as the result of the measurement of \(\overline{P}_M\).

To analyze the fault tolerance of such a protocol, we follow~\cite{cross2024improved} and consider the phenomenological (or spacetime) fault distance, which is the minimum number of single qubit errors and measurement errors required to cause a logical error. We can further separate logical errors into \emph{spacelike} errors, which commute with the logical operator being measured (\textit{i.e.} an error that affects the logical information not being measured), and \emph{timelike} errors, which flip the measurement result. We can see that timelike errors must flip an odd number of vertex checks in \(\supp(c)\) during the first round of syndrome measurement of the merged code.

Let \(\mathcal{S}\) be the stabilizer of the base code (including single qubit \(X\) operators on the edge qubits of the ancilla system), \(\{\overline{P}_1, \overline{Q}_1, \dots, \overline{P}_{k-1}, \overline{Q}_{k-1}, \overline{P}_M, \overline{Q}_M\}\) be a symplectic basis of logical operators, and \(\mathcal{S}'\) be the stabilizer of the deformed code.

Note that we can ensure that the set \(\mathcal{L} = \langle \overline{P}_1, \overline{Q}_1, \dots, \overline{P}_{k-1}, \overline{Q}_{k-1}\rangle\) forms a valid set of logical operators for both \(\mathcal{S}\) and \(\mathcal{S}'\). To see this, let \(\overline{R}\) be a logical operator of the base code that commutes with \(\overline{P}_M\). The only checks in \(\mathcal{S}'\) that could anticommute with \(\overline{R}\) are the vertex checks, and \(\{\overline{R}, S\} = 0\) for a vertex check \(S\) if and only if the set \(\{ q \in f^{-1}(S) : \{\overline{R}(q), \overline{P}_M(q)\} = 0\}\) has an odd size. The total set \(\{q : \{\overline{R}(q), \overline{P}_M(q)\} = 0\}\) must have an even size, and thus \(\overline{R}\) must anticommute with an even number of vertex checks. Let \(\mu\) be a path matching of these vertices in the ancilla graph. We can see that \(\overline{R} \prod_{e \in \mu} X(e)\) is a valid logical operator for both \(\mathcal{S}\) and \(\mathcal{S}'\).

Let \(\mathcal{L}^* = \mathcal{L} \setminus \{I\}\), and \(\mathcal{G} = \langle \mathcal{S}, \mathcal{S}'\rangle\). We then have the following:

\begin{lemma}[Lemma 9 of \cite{cross2024improved}]
  \label{thm:fault-dist}
 The timelike fault distance of the protocol described above is \(d(\overline{Q}_M\mathcal{L} \mathcal{G})\), and the spacelike fault distance is \(d(\mathcal{L}^* \mathcal{G})\).
\end{lemma}

When the base code is CSS, and the logical operator being measured is a \(Z\) type or \(X\) type logical operator, we may take our symplectic basis to be CSS as well, \emph{e.g.} \(\mathcal{L} = \langle \overline{X}_1, \overline{Z}_1, \dots, \overline{X}_{k-1}, \overline{Z}_{k-1}\rangle\), and we can further split up our stabilizer into the set of \(X\) checks \(\mathcal{S}_X\) and \(Z\) checks \(\mathcal{S}_Z\). In this case, we have the following further simplification of the fault distance, as originally noted by \cite{cross2024improved}.

\begin{theorem}[Theorem 11 of \cite{cross2024improved}, slightly modified]
  \label{thm:css-fault-dist}
  When \(\mathcal{Q}\) is a CSS code of distance \(d\), and \(\overline{P}_M\) is a \(Z\) type logical operator, then the fault distance of the protocol described above is \(\min (d, d_Z(\mathcal{L}^* \mathcal{S}'))\), \textit{i.e.} the minimum of \(d\) and the \(Z\)-distance of the merged code. When \(\overline{P}_M\) is an \(X\) type logical operator, then the fault distance is \(\min (d, d(\mathcal{L}_X^* \mathcal{S}_V' \mathcal{S}_X))\) where \(\mathcal{S}'_V\) is the group generated by vertex checks of the merged code, and \(\mathcal{L}_X^* = \langle \overline{X}_1, \dots, \overline{X}_{k-1}\rangle \setminus \{I\}\).
\end{theorem}

\begin{proof}
  The case of measuring a \(Z\) type logical is identical to Theorem 11 of \cite{cross2024improved}. However, due to our definition of an ancilla system, when measuring an \(X\) type logical operator on a CSS code, our resulting merged code is no longer CSS, and the argument of \cite{cross2024improved} no longer applies directly. In this case though, we may consider an identical protocol applied to a base code with all of the \(X\) and \(Z\) operators swapped. In this protocol, the operator being measured is \(Z\) type, and we can apply the result of Theorem 11 of \cite{cross2024improved}. The relevant \(Z\) type operators in this exchanged basis correspond to the set \(\mathcal{L}_X^* \mathcal{S}_V' \mathcal{S}_X\) which yields the desired expression for the distance.
\end{proof}

\section{Extractor Construction}
\label{sec:extractor-construction}

In this section, we provide the details behind our extractor construction, and prove its fault tolerance.
The construction follows three steps. First, we construct extractors capable of measuring any logical operator in a given information row/column of an HGP code (Appendix \ref{sec:single-col}).
Then, we connect several of these systems together using additional qubits and checks called \emph{bridges} \cite{swaroop2024universal,cross2024improved}, to obtain an extractor capable of measuring every logical operator of a given basis (Appendix \ref{sec:single-basis-ext}). 
Finally, we connect two such extractors together with another bridge to obtain an extractor capable of measuring every logical operator in the code (Appendix \ref{sec:full-ext}). 
The fault tolerance of the full extractors relies on the fault tolerance of the single-column extractors, which can be verified numerically.
Providing provable constructions for such single-column extractor graphs remains an interesting area for future research. Figure~\ref{fig:extractor-construction-tikz} gives an overview of all three levels of the construction.

%
%
%

\begin{figure*}[tbp]
\centering
\begin{tikzpicture}[
  font=\footnotesize,
  every node/.style={inner sep=2pt},
]

\fill[RoyalBlue!10] (0, 0.75) rectangle (16, -4.15);
\draw[RoyalBlue!55, thick, rounded corners=4pt] (0, 0.75) rectangle (16, -4.15);
\node[font=\small\bfseries] at (8, 0.465)
  {Single Column Extractor \((\mathcal{X}_C,\, f,\, D)\)\quad[\S\ref{sec:single-col}]};

\fill[RoyalBlue!22] (0.20, 0.18) rectangle (5.10, -3.95);
\draw[RoyalBlue!55, thin, rounded corners=2pt] (0.20, 0.18) rectangle (5.10, -3.95);
\node[font=\footnotesize\bfseries, anchor=north] at (2.65, 0.15)
  {Ancilla graph $G=(V,E)$};

\def\vrad{8pt}
\coordinate (V1) at (1.35, -1.05);
\coordinate (V2) at (4.05, -1.05);
\coordinate (V3) at (1.35, -2.75);
\coordinate (V4) at (4.05, -2.75);

\draw[thick] (V1) -- (V2);
\draw[thick] (V1) -- (V3);
\draw[thick] (V2) -- (V4);
\draw[thick] (V3) -- (V4);
\draw[thick] (V1) -- (V4);  

\filldraw[fill=RoyalBlue!50, draw=RoyalBlue!80] (V1) circle (\vrad);
\node[font=\scriptsize, left=6pt] at (V1) {$v_1$};
\filldraw[fill=RoyalBlue!50, draw=RoyalBlue!80] (V2) circle (\vrad);
\node[font=\scriptsize, right=6pt] at (V2) {$v_2$};
\filldraw[fill=RoyalBlue!50, draw=RoyalBlue!80] (V3) circle (\vrad);
\node[font=\scriptsize, left=6pt] at (V3) {$v_3$};
\filldraw[fill=RoyalBlue!50, draw=RoyalBlue!80] (V4) circle (\vrad);
\node[font=\scriptsize, right=6pt] at (V4) {$v_4$};

\node[font=\scriptsize, inner sep=1pt] at (2.70, -0.90) {$e_1$};
\node[font=\scriptsize, inner sep=1pt] at (1.20, -1.90) {$e_2$};
\node[font=\scriptsize, inner sep=1pt] at (4.25, -1.90) {$e_3$};
\node[font=\scriptsize, inner sep=1pt] at (2.70, -2.60) {$e_4$};
\node[font=\scriptsize, inner sep=1pt] at (3.05, -1.95) {$e_5$};

\node[font=\scriptsize, anchor=north, align=center] at (2.65, -3.00)
  {vertices $\leftrightarrow$ $Z$-checks\\
   edges $\leftrightarrow$ ancilla qubits\\
   cycles $\leftrightarrow$ $X$-checks};

\fill[RoyalBlue!8] (5.30, 0.18) rectangle (10.25, -3.95);
\draw[RoyalBlue!45, thin, rounded corners=2pt] (5.30, 0.18) rectangle (10.25, -3.95);
\node[font=\footnotesize\bfseries, anchor=north] at (7.78, 0.15)
  {Ancilla $\mathcal{X}_C=\mathrm{CSS}(H_X,H_Z)$};

\node[anchor=north west, align=left, font=\scriptsize] at (5.45, -0.850) {%
  $H_Z$: incidence matrix of $G$\\[4pt]
  $H_X$: rows span $\ker(H_Z)$ {\tiny (cycle basis)} \\[4pt]
  $k_{\mathcal{X}}=0$;\quad $\dim\ker(H_Z^T)=1$\\[8pt]
  \textbf{Port function} $f: q_{i,1}\mapsto v_i$\\[4pt]
  \textbf{Deformation map} $D:\mathcal{S}\!\to\!\mathcal{P}(E)$\\[4pt]
};

\fill[RoyalBlue!8] (10.45, 0.18) rectangle (15.80, -3.95);
\draw[RoyalBlue!45, thin, rounded corners=2pt] (10.45, 0.18) rectangle (15.80, -3.95);
\node[font=\footnotesize, anchor=north, align=center] at (13.13, 0.15)
  {{\bf Measures}\\{\tiny (one column or row at a time)}};

\node[anchor=north west, align=left, font=\scriptsize] at (10.60, -0.95) {%
  $Z(c\cdot e_j^T,0)$,\; $c\!\in\!\ker\partial_A$,\;
    $j\!\in\![k_B]$\\[4pt]
  $X(e_i\cdot c^T,0)$,\; $i\!\in\![k_A]$,\;
    $c\!\in\!\ker\partial_B$\\[4pt]
  $Z(0,e_i\!\cdot\!\tilde{c}^T)$,\; $\tilde{c}\!\in\!\ker\partial_B^T$\\[4pt]
  $X(0,\tilde{c}\!\cdot\!e_j^T)$,\; $\tilde{c}\!\in\!\ker\partial_A^T$\\[5pt]
};

\fill[ForestGreen!10] (0, -4.55) rectangle (16, -9.30);
\draw[ForestGreen!55, thick, rounded corners=4pt] (0, -4.55) rectangle (16, -9.30);
\node[font=\small\bfseries] at (8, -4.825)
  {Single Basis Extractor ($\mathcal{X}_Z,\,f_Z,\,D_Z$)\quad[\S\ref{sec:single-basis-ext}]};

\fill[ForestGreen!18] (0.20, -5.10) rectangle (9.75, -9.10);
\draw[ForestGreen!50, thin, rounded corners=2pt] (0.20, -5.10) rectangle (9.75, -9.10);
\node[font=\footnotesize\bfseries, anchor=north] at (4.98, -5.08)
  {$k$ copies of $G$, joined by bridges $B^{(i)}$};

\fill[RoyalBlue!18] (0.45, -5.60) rectangle (2.65, -8.52);
\draw[RoyalBlue!55, thin, rounded corners=2pt] (0.45, -5.60) rectangle (2.65, -8.52);
\node[font=\scriptsize\bfseries, anchor=north] at (1.55, -6.20) {$G^{(1)}$};
\node[anchor=north west, align=left, font=\scriptsize] at (0.57, -6.60) {%
  $|V|$ $Z$-checks\\[3pt]
  $|E|$ qubits\\[5pt]
  $f_Z(i,1)\!=\!v_i^{(1)}$
};

\foreach \y in {-6.40, -7.06, -7.72} {
  \draw[<->, thin, BrickRed!70] (2.75, \y) -- (3.75, \y);
}
\node[font=\scriptsize, BrickRed!80, anchor=north] at (3.25, -7.85) {$B^{(1)}$};

\fill[RoyalBlue!18] (3.85, -5.60) rectangle (6.05, -8.52);
\draw[RoyalBlue!55, thin, rounded corners=2pt] (3.85, -5.60) rectangle (6.05, -8.52);
\node[font=\scriptsize\bfseries, anchor=north] at (4.95, -6.20) {$G^{(2)}$};
\node[anchor=north west, align=left, font=\scriptsize] at (3.97, -6.60) {%
  $|V|$ $Z$-checks\\[3pt]
  $|E|$ qubits\\[5pt]
  $f_Z(i,2)\!=\!v_i^{(2)}$
};

\node[font=\small, anchor=center] at (6.68, -7.06) {$\cdots$};

\fill[RoyalBlue!18] (7.30, -5.60) rectangle (9.50, -8.52);
\draw[RoyalBlue!55, thin, rounded corners=2pt] (7.30, -5.60) rectangle (9.50, -8.52);
\node[font=\scriptsize\bfseries, anchor=north] at (8.40, -6.20) {$G^{(k)}$};
\node[anchor=north west, align=left, font=\scriptsize] at (7.42, -6.60) {%
  $|V|$ $Z$-checks\\[3pt]
  $|E|$ qubits\\[5pt]
  $f_Z(i,k)\!=\!v_i^{(k)}$
};

\node[font=\scriptsize, anchor=north] at (4.98, -8.50)
  {right block: $f_Z(\sigma(i),\!\sigma(j))\!=\!v_j^{(i)}$\;
   (permutation $\sigma$ on $[n]$, see text)};

\fill[ForestGreen!8] (9.95, -5.10) rectangle (15.80, -9.10);
\draw[ForestGreen!45, thin, rounded corners=2pt] (9.95, -5.10) rectangle (15.80, -9.10);
\node[font=\footnotesize\bfseries, anchor=north] at (12.88, -5.13)
  {Properties};

\node[anchor=north west, align=left, font=\scriptsize] at (10.10, -5.73) {%
  Measures \emph{any} $Z$-basis logical operator:\\[4pt]
  $\Zlog{=}Z\!\bigl(\textstyle\sum_i c_i{\cdot}e_i^T,\,
    \sum_j e_{\sigma(j)}{\cdot}\tilde{c}_j^T\bigr)$\\[8pt]
  \textbf{Theorem~\ref{thm:single-basis-dist}}: if $(\mathcal{X},\,f,\,D)$ is\\[4pt]
  distance-preserving, then $(\mathcal{X}_Z,\,f_Z,\,D_Z)$\\[4pt]
  is distance-preserving.\\[8pt]
  Analogously: $(\mathcal{X}_X,\,f_X,\,D_X)$ for $X$-basis
};

\fill[Mahogany!10] (0, -9.70) rectangle (16, -14.40);
\draw[Mahogany!55, thick, rounded corners=4pt] (0, -9.70) rectangle (16, -14.40);
\node[font=\small\bfseries] at (8, -10.05)
  {Full Extractor $(\mathcal{X},\,f_F,\,D_F)\quad[\S\ref{sec:full-ext}]$};

\fill[Mahogany!15] (0.20, -10.30) rectangle (9.75, -14.30);
\draw[Mahogany!50, thin, rounded corners=2pt] (0.20, -10.30) rectangle (9.75, -14.30);

\fill[BrickRed!18] (0.40, -11.20) rectangle (4.20, -13.90);
\draw[BrickRed!55, thin, rounded corners=2pt] (0.40, -11.20) rectangle (4.20, -13.90);
\node[font=\scriptsize\bfseries, anchor=north] at (2.30, -11.70)
  {$X$-extractor $(\mathcal{X}_X,\,f_X,\,D_X)$};
\node[anchor=north, align=center, font=\scriptsize] at (2.15, -12.10) {%
  $G_X$: $k$ copies of $G$\\[4pt]
  $f_X\!:\!\supp(\Xlog)\!\to\!V_X$\\[4pt]
  Measures $X$-basis ops.
};

\fill[Mahogany!25] (4.30, -11.20) rectangle (5.65, -13.90);
\draw[Mahogany!55, thin] (4.30, -11.20) rectangle (5.65, -13.90);
\node[font=\scriptsize, Mahogany!90, rotate=0] at (4.95, -12.15)
  {$k^2$ };
\node[font=\scriptsize, Mahogany!90, rotate=0] at (4.95, -12.55)
  {bridge};
\node[font=\scriptsize, Mahogany!90, rotate=0] at (4.95, -12.95)
  {edges};

\fill[RoyalBlue!18] (5.75, -11.20) rectangle (9.55, -13.90);
\draw[RoyalBlue!55, thin, rounded corners=2pt] (5.75, -11.20) rectangle (9.55, -13.90);
\node[font=\scriptsize\bfseries, anchor=north] at (7.60, -11.70)
  {$Z$-extractor $(\mathcal{X}_Z,\,f_Z,\,D_Z)$};
\node[anchor=north, align=center, font=\scriptsize] at (7.60, -12.10) {%
  $G_Z$: $k$ copies of $G$\\[4pt]
  $f_Z\!:\!\supp(\Zlog)\!\to\!V_Z$\\[4pt]
  Measures $Z$-basis ops.\\
};

\node[font=\scriptsize, anchor=north, align=center] at (4.98, -13.90)
  {$f_F(q)=f_Z(q)$ if $q\in\mathrm{dom}(f_Z)$,\; else $f_X(q)$};

\node[font=\footnotesize\bfseries, anchor=north, align=center] at (4.98, -10.35)
  {$X$-extractor $\mathcal{X}_X$ and $Z$-extractor $\mathcal{X}_Z$\\
   joined by $k^2$ bridge edges};

\fill[Mahogany!8] (9.95, -10.30) rectangle (15.80, -14.30);
\draw[Mahogany!45, thin, rounded corners=2pt] (9.95, -10.30) rectangle (15.80, -14.30);
\node[font=\footnotesize\bfseries, anchor=north] at (12.88, -10.35)
  {Properties};

\node[anchor=north west, align=left, font=\scriptsize] at (10.10, -10.90) {%
  Measures all logical Pauli operators\\[4pt]
  $\overline{P}_M = \Xlog\cdot\Zlog$\\[8pt]
  \textbf{Fault distance}:\\[4pt]
  $\min\!\bigl(d,\;k^2,\;k\cdot\lambda(G)\bigr)$\\[3pt]
  $\lambda(G)$: edge connectivity of~$G$\\[6pt]
  $\mathcal{X}$ adds $2|E_X| + 2|E_Z| + 2k^2 + 1$\\
  ~~ ancilla qubits to the code 
};

\end{tikzpicture}
\caption{
  Assembly of the full extractor for a cyclic
  $\HGP(\partial,\partial)$ code
  (Appendix~\ref{sec:extractor-construction}).
  \emph{Top:} A single-column extractor $(\mathcal{X}_C,f,D)$ is built from a graph $G=(V,E)$
  with $|V|=n$: the incidence matrix of $G$ gives the ancilla $Z$-checks, a
  cycle basis of $G$ gives the $X$-checks, a port function $f$ maps code
  qubits to vertices, and a deformation map $D:\mathcal{S}\to\mathcal{P}(E)$ assigns each
  base-code stabilizer a path matching in $G$.  The HGP surgery subcode
  (Theorem~\ref{thm:sur-sc-dist}) reduces the associated distance
  computation from $O(n^2)$ to $O(n + |E| )$ qubits.
  \emph{Middle:} A single-basis extractor $(\mathcal{X}_Z,f_Z,D_Z)$ stacks $k$
  copies of $G$ linked by bridges (consisting of edges and additional cycles) and measures any canonical
  $Z$-logical operator; it is distance-preserving whenever the
  single-column extractor is (Theorem~\ref{thm:single-basis-dist}).  An
  $X$-basis extractor $(\mathcal{X}_X,f_X,D_X)$ is constructed analogously.
  \emph{Bottom:} The full extractor $(\mathcal{X},f_F,D_F)$ joins
  $\mathcal{X}_X$ and $\mathcal{X}_Z$ via $k^2$ additional bridge edges and
  measures any logical Pauli operator $\overline{P}_M=\Xlog\cdot\Zlog$ with
  fault distance $\min(d,k^2,k\cdot\lambda(G))$, where $\lambda(G)$ is the
  edge connectivity of~$G$.}
\label{fig:extractor-construction-tikz}
\end{figure*}

Throughout this section, \(\mathcal{S}\) denotes the stabilizer of the base code, now including the single qubit \(X\) operators on the qubits of the ancilla system, and \(\mathcal{S}'\) denotes the stabilizers of the merged code. Since all of the base codes will be HGP codes, which are CSS, we let \(\mathcal{S}_X\) and \(\mathcal{S}_Z\) denote the \(X\) and \(Z\) checks of the base code.
Let \(\mathcal{S}'_V\) denote the group generated by the vertex checks of the merged code, and \(\mathcal{S}_C\) denote the group generated by the cycle checks of the merged code.

\subsection{Single Column/Row Extractors}
\label{sec:single-col}

The basis of our extractor construction is a ``single-column'' extractor, \textit{i.e.} an ancilla system capable of measuring any operator of the forms
\begin{enumerate}
\item \(Z(c \cdot e_j^\top , 0)\), for \(c \in \ker{\partial_A}, j \in [k_B]\)
\item \(X(e_i \cdot c^\top , 0)\), for \(i \in [k_A], c \in \ker{\partial_B}\)
\item \(Z(0, e_i \cdot \tilde{c}^\top )\), for \(i \in [k_A^\top ], \tilde{c} \in \ker{\partial_B^\top }\)
\item \(X(0, \tilde{c} \cdot e_j^\top )\), for \(\tilde{c} \in \ker{\partial_A^\top }, j \in [k_B^\top ]\).
\end{enumerate}
where the specific set of operators is determined by which row or column the extractor is connected to, and
where we are assuming that we have ordered the bits of the classical codes in such a way that the first \(k_A\), \(k_B\), \(k_A^\top \), and \(k_B^\top \) indices form information sets for \(\mathcal{C}_A, \mathcal{C}_B, \mathcal{C}_A^\top \), and \(\mathcal{C}_B^\top \) respectively. 
For simplicity we will focus on logical operators of the first form, but all of our arguments extend straightforwardly to the rest.
By the same argument as in Appendix~\ref{sec:meas}, we may take \(\mathcal{L} = \langle \overline{X}_1, \overline{Z}_1, \dots , \overline{X}_{k-1}, \overline{Z}_{k-1} \rangle\) to be the subgroup generated by unmeasured logical operators, where we take \(\{\overline{X}_1, \overline{Z}_1, \dots, \overline{X}_M, \overline{Z}_M\}\) to be a symplectic basis of logical operators for \(\mathcal{S}\), and we assume that we are measuring the operator \(\overline{Z}_M\).

By Theorem \ref{thm:css-fault-dist}, the fault distance of using an extractor to measure such a logical operator is given by the minimum of \(d\) and
\begin{equation}
d \left ( \mathcal{L}_Z^* \mathcal{S}_Z \mathcal{S}_V' \right )
\end{equation}
where \(\mathcal{L}_Z^* = \langle \Zlog_1, \dots, \Zlog_{k-1}\rangle \setminus \{I\}\).

We now demonstrate that we only have to look at a restricted set of logical \(Z\) operators when analyzing the distance of single-column systems.
\begin{lemma}
  \label{lem:single-col}
	Let $\Q$ be an HGP code with distance \(d\), and \((\A, f, D)\) be an ancilla system used to measure the logical operator $\overline{Z} = Z(c \cdot e_j^\top , 0), c \in \ker{\partial_A} \setminus \{0 \}$ supported entirely on a single column \(j \in [k_B]\). Then, the fault distance during the measurement procedure is at least the minimum of \(d\) and
\begin{equation}
    d \left (
      \left \{ Z(c' \cdot e_j^\top , 0) \Lambda_A : c' \in \ker{\partial_A} \setminus \{c, 0\}, \Lambda_A \in \mathcal{S}'_V  \right \}
    \right )
\end{equation}
\end{lemma}

\begin{proof}
  We will show that for any logical \(Z\) operator \(\Zlog'\) in a different homology class than \(\Zlog\), and any \(\Lambda_A \in \mathcal{S}'_V\), either \(\wt(\Zlog' \Lambda_A)\ge d\), or there exists a \(\Zlog''\) contained entirely in the same column as \(\Zlog\) such that \(\wt(\Zlog' \Lambda_A) \ge \wt(\Zlog'' \Lambda_A)\).
  Without loss of generality we can take \(\Zlog = Z(c \cdot e_1^\top , 0)\), and \(\Zlog' = Z(L + U\partial_B, R + \partial_A U)\) for \(U \in \mathbb{F}_2^{n_A \times m_B}\), \(L = \sum_{j=1}^{k_B} c_j \cdot e_j^\top \), and \(R = \sum_{i = 1}^{k_A^\top } e_i \cdot \tilde{c}_i^\top \), where \(c_j \in \ker{\partial_A}\) and \(\tilde{c}_i \in \ker{\partial_B^\top }\).

  By Lemma \ref{lem:z-hom-inv}, if \(R \ne 0\) then \(\wt(\Zlog' \Lambda_A) \ge d_B^\top  \ge d\).
  Thus we may restrict our attention to the case where \(R = 0\).
  Additionally, if there exists a \(j' > 1\) such that \(c_{j'} \ne 0\), then by Lemma \ref{lem:hgp-clean}, \(\wt(\Zlog' \Lambda_A) \ge d_A \ge d\).
  Thus we can assume that \(L = c' \cdot e_1^\top \).
  Let \(\eta\) be the set of vertices corresponding to \(\Lambda_A\), and let \(S_A\) be the support of \(\Lambda_A\) on the ancilla system.
  We then have that
  \begin{equation}
    \wt(\Zlog' \Lambda_A) = |(c' + f^{-1}(\eta))e_1^\top  + U\partial_B| + |\partial_A U| + |S_A|.
  \end{equation}
  Since \(e_1 \notin \rs(\partial_B)\), for any \(i\) such that \((U\partial_B)_{i, 1} = 1\), there exists a \(j > 1\) such that \((U\partial_B)_{i, j} = 1\). Thus,
  \begin{equation}
    |(c' + f^{-1}(\eta)) e_1^\top  + U\partial_B| \ge |(c' + f^{-1}(\eta))e_1^\top |.
  \end{equation}
  Letting \(\Zlog'' = Z(c'\cdot e_1^\top , 0)\), we see that \(\wt(\Zlog' \Lambda_A) \ge \wt(\Zlog'' \Lambda_A)\) as desired.
\end{proof}

We now show that to determine the distance of the deformed code, it suffices to examine a much smaller code that we call the \emph{HGP surgery subcode}.

\begin{definition}[HGP Surgery Subcode]
Let $\mathcal{Q} = \mathrm{HGP}(\partial_A, \partial_B)$ be an HGP code, where \(\partial_A \in \mathbb{F}_2^{m_A \times n_A}\), and let $\overline{Z} = Z(c \cdot e_j^\top , 0)$ for $c \in \ker{\partial_A} \setminus \{0\}$ and $j \in [k_B]$ be a logical $Z$ operator. Let $\A = \mathrm{CSS}(H_X, H_Z)$ be an ancilla system with $\hat{n}$ qubits, $\hat{m}_{X}$ $X$ checks, $\hat{m}_{Z}$ $Z$ checks, \(k_{\A} = 0\), an injective port function $f$, and a deformation map \(D\).

The \emph{HGP surgery subcode} is a CSS code defined as follows: 
\begin{itemize}
	\item The code has $n_A + \hat{n}$ qubits that we break into two sets of size $n_A$ and $\hat{n}$. We call these sets the \emph{base block} and \emph{ancilla block} respectively.
	\item For each $i \in [\hat{m}_{Z}]$ we add a $Z$ check with support on $H_Z^\top  e_i$ on the ancilla block, as well as \(f^{-1}(\{i\})\) in the base block.
	\item For each $i \in [m_A]$ we add an $X$ check \(S\), which has support $\partial_A^\top  e_i$ on the $n_A$ block of qubits, and \(D(S)\) on the ancilla block, where \(D(S)\) by definition is a path matching of \(f(c \cap \partial_A^\top  e_i)\).
	\item For each $i \in [\hat{m}_{X}]$, we add an $X$ check with support on $H_X^\top  e_i$ on the ancilla block. 
\end{itemize}
\end{definition}

\begin{theorem}
  \label{thm:sur-sc-dist}
Suppose $\partial_A$ is a parity check matrix, and let $c \in \ker{\partial_A} \setminus \{0\}$.
Then the HGP surgery subcode associated with $\mathrm{HGP}(\partial_A, \partial_B)$ and the logical $Z$ operator $Z(c \cdot e_j^{\top}, 0)$ for $j \in [k_B]$ has the same $Z$ distance as the deformed code formed by joining $\HGP(\partial_A, \partial_B)$ and an ancilla system used to measure $\overline{Z}$.
\end{theorem}

\begin{proof}
  Consider the set of \(Z\) operators
  \begin{equation}
    \mathcal{L}_Z^* = \{ Z(c'\cdot e_j^\top , 0) \Lambda_A : c' \in \ker{\partial_A} \setminus \{c, 0\}, \Lambda_A \in \mathcal{S}'_V \}
  \end{equation}
  By Lemma \ref{lem:single-col}, we can see that for the deformed HGP code,
  \begin{equation}
    d_Z = d ( \mathcal{L}_Z^*  )
  \end{equation}

  We note that there are \(k_A - 1\) independent logical \(Z\) operators in the set \(\mathcal{L}_Z^*\) (up to multiplication by vertex checks), and that each operator is also a logical operator of the HGP surgery subcode. To prove the theorem, it then suffices to show that these are \emph{all} of the logical $Z$ operators of the HGP surgery subcode, \emph{i.e.} it encodes $k_A - 1$ logical qubits.

Consider that the HGP surgery subcode has $n_A + \hat{n}$ qubits. We now count the number of independent checks.
Note that without considering their support on the ancilla block, \(n_A - k_A\) of the \(m_A\) base code \(X\) checks are independent. Suppose there is a subset of checks, represented by \(\eta \in \mathbb{F}_2^{m_A}\), which is linearly dependent, \emph{i.e.} \(\partial_A^\top  \eta = 0\). Their product will have support on the ancilla block corresponding to a path matching of \(f(c \cap \partial_A^\top \eta)\). However, since they are dependent, \(\partial_A^\top \eta = 0\), and their support on the ancilla block must be equal to an element of \(\ker(H_Z)\). Since \(k_{\A} = 0\), \(\rs(H_X) = \ker(H_Z)\) and we see that the support of these checks is equal to the product of cycle checks, and we still have a linearly dependent set.
Suppose \(\A\) has \(\hat{m}'_Z\) independent \(Z\) checks and \(\hat{m}'_X\) independent \(X\) checks. Since \(k_{\A} = 0\), \(\hat{m}'_X + \hat{m}_Z' = \hat{n}\).
Since $\nullity(H_Z^\top ) = 1$, there is only one product of $Z$ checks that does not have support on the ancilla block, which is the product of all of the $Z$ checks corresponding to the support of \(0 \ne \hat{c} \in \ker(H_Z^\top )\).
However, this will have support on the $n_A$ block, and so the $Z$ checks are independent. Thus, we can see that in total, the number of independent checks is \(n_A - k_A + \hat{m}'_X + \hat{m}'_Z + 1 = n_A - k_A + \hat{n} + 1\), and thus the HGP surgery subcode encodes \(k_A - 1\) logical qubits.
\end{proof}

\begin{remark}
    A variety of methods for computing the distance of quantum codes exist \cite{webster2026distancefinding}. In the case of HGP codes, if the seed classical codes have length \(n\) bits, then the HGP code has \(O(n^2)\) qubits.
    Thus, Theorem~\ref{thm:sur-sc-dist} reduces the size of the code passed to the solver from \(O(n^2)\) to \(O(n + |\mathcal{X}_C|)\), where \(|\mathcal{X}_C|\) is the number of data qubits in the single-column extractor.
    For all of the constructions in this work, we found \(|\mathcal{X}_C| = 2n\).
    Using the CP-SAT solver \cite{cpsatlp} from the OR-TOOLS \cite{ortools} software package, we were able to verify the distance for each of the systems introduced in this work in several minutes on a consumer laptop by computing the distance of logical \(Z\) operators on the corresponding HGP surgery subcodes.
\end{remark}

With the above results in mind, we now describe our approach for constructing single-column extractors \((\mathcal{X}_C, f, D)\).
For simplicity suppose we have an HGP code \(\HGP(\partial, \partial)\) with \(\partial \in \mathbb{F}_2^{m \times n}\).
We begin with a connected graph \(G = (V, E)\), where \(|V| = n\).
For \(\mathcal{X}_C\), let \(H_Z \in \mathbb{F}_2^{|V| \times |E|}\), with \((H_Z)_{i,j} = 1\) if \(v_i \in e_j\), and \((H_Z)_{i,j} = 0\) otherwise.
For each \(X\) check \(S\) in the first column of the HGP code, we assign \(D(S)\) to a set of edges whose edge-induced subgraph contains a path matching of \(K(S)\) for each logical \(Z\) operator supported in the first column.
Let \(H_X\) be a matrix whose rows span \(\ker(H_Z)\) --- the rows of \(H_X\) correspond to a cycle basis for the graph \(G\).
For \(f\), assign some ordering of the vertices of \(V\), \emph{i.e.} \(V = \{v_1, v_2, \dots, v_{n}\}\) and take \(f: [n] \to V\) to be \(f(i) = v_i\).
This results in an extractor \(\mathcal{X}_C\) with \(|V|\) \(Z\) checks, \(|E|\) data qubits, and \(|E| - |V| + 1\) \(X\) checks, and thus \(2|E| + 1\) total qubits.

\begin{remark}
    In practice, to construct the graph \(G\), we begin with \(|V| = n\) and \(E = \emptyset\), and then for each \(X\) check \(S\), add edges forming a path through \(f^{-1}(\mathrm{supp}(S))\). For each code, we attempted this method using several different orderings of the \(X\) checks, and took the graph with the smallest maximum degree.
\end{remark}

We list the degrees of each kind of qubit in the construction in Table \ref{tab:single-column}.
The table uses notation we define now.
Let \(w\) be the maximum weight of a row of \(\partial\), \(r\) be the maximum weight of a column of \(\partial\), \(\Delta_Z\) be the maximum weight of the rows of \(H_Z\), \(\Delta_X\) be the maximum weight of a row of \(H_X\), \(\rho\) be the maximum weight of a column of \(H_X\).
Note that the maximum weight of a column of \(H_Z\) is two.
For each codeword \(c_i \in \ker{\partial}\), and each \(j \in [m]\), let \(\mu_{i,j}\) be a path matching of \(f(c_i \cap \partial^\top  e_j)\) (which must exist, because \(G\) is a graph, \(f\) is injective, and \(|c_i \cap \partial^\top  e_j|\) is even), and let \(u = \max_j | \cup_{i} \mu_{ij}|\).
Finally, for each edge \(e \in E\), let \(p_e = | \{ j : \exists i \text{ such that } e \in \mu_{ij} \} |\), and let \(p = \max_{e \in E} p_e\). 

\begin{table}[h]
    \centering
    \begin{tabular}{ccc}
    \toprule
    Qubit Type & Max. Degree & \makecell{Example Value for \\ \(d=10\) System}\\
    \midrule
    Left Block Data Qubit  & \(2r + 1\) & \(7\) \\
    Right Block Data Qubit & \(2w\) & \(6\) \\
    Base Code \(X\) Check & \(w + r + u\) & 8 \\
    Ancilla System Data Qubit & \(2 + \rho + p\) & 6 \\
    Vertex Check & \(\Delta_Z + 1\) & 5 \\
    Cycle Check & \(\Delta_X\) & 6 \\
    \bottomrule
    \end{tabular}
    \caption{Qubit degrees in single-column extractor system. The last column lists the values for the extractor system for the \([[882,50,10]]\) cyclic HGP code.}
    \label{tab:single-column}
\end{table}

\subsection{Single Basis Extractors}
\label{sec:single-basis-ext}

In this section, we show how to combine the single row/column extractors described in the previous section into extractors capable of measuring any logical \(X\) or \(Z\) operator, and prove that these constructions are distance preserving. For the rest of the appendices, unless otherwise specified, we will only consider cyclic HGP codes (described in Section \ref{sec:cyclic-hgp}). Again, for simplicity, we will only consider extractors for measuring logical \(Z\) type operators, but the arguments extend straightforwardly to \(X\) type operators as well.

Suppose we are using an \([n, k, d]\) classical cyclic code. Let \((\mathcal{X}_C, f, D)\) be a distance preserving single-column ancilla system as described in the previous section, with \(\mathcal{X}_C = \mathrm{CSS}(H_X, H_Z)\).
Let \(G = (V, E)\) be the graph used to construct \(\mathcal{X}_C\).
Our \(Z\) basis extractor \((\mathcal{X}_Z, f_Z, D_Z)\) will be based off of the graph \(G_Z\) obtained by taking \(k\) copies of \(G\), and adding edges between corresponding vertices in copies \(i\) and \(i + 1\).
More concretely, \(G_Z = (V_Z, E_Z)\), with \(V_Z = \cup_{i = 1}^k V^{(i)}\), \(E_Z = \left ( \cup_{i=1}^{k} E^{(i)} \right) \cup \left ( \cup_{i=1}^{k-1} B^{(i)} \right )\), where for each \(i \in [k]\) \(|V^{(i)}| = |V|\), \((v^{(i)}_j, v^{(i)}_{j'}) \in E^{(i)}\) if and only if \((v_j, v_{j'}) \in E\), and for each \(i \in [k-1]\), \(B^{(i)} \subset V^{(i)} \times V^{(i+1)}\) such that \( B^{(i)} = \{ (v^{(i)}_j, v^{(i+1)}_j) : j \in [|V|] \}\).
We call the qubits associated with the edges in \(\cup_{i=1}^{k-1} B^{(i)}\) \emph{bridge qubits}.
An example \(Z\) parity check matrix corresponding to \(G_Z\) for the \(k=3\) case is shown in Equation \eqref{eq:z-pcm}.

\begin{equation}
\label{eq:z-pcm}
\bordermatrix{
  & E_1 & B_1 & E_2 & B_2 & E_3 \cr
V_1 & H_Z & I   & 0   & 0   & 0   \cr
V_2 & 0   & I   & H_Z & I   & 0   \cr
V_3 & 0   & 0   & 0   & I   & H_Z
}
\end{equation}

To determine the \(X\) parity check matrix, we must identify a cycle basis for \(G_Z\). Note that a simple cycle basis is obtained by taking the cycles corresponding to \(G^{(1)} = (V^{(1)}, E^{(1)})\), and then the \((k-1)\cdot |E|\) length four cycles of the form 
\[
[(v^{(i)}_j, v^{(i)}_{j'}), (v^{(i)}_{j'}, v^{(i+1)}_{j'}), (v^{(i+1)}_{j'}, v^{(i+1)}_{j}), (v^{(i+1)}_{j}, v^{(i)}_{j})]
\] 
for \(i \in [k-1]\) and \((v_j, v_{j'}) \in E\).

Having defined an \(\mathcal{X}_Z\) with \(k_{\mathcal{X}_Z} = 0\), it remains to define \(f_Z\). Given qubit \((i, j)\) in the left block, with \(j \in [k]\), we let \(f(i,j) = v_i^{(j)}\). To define \(f_Z\) on the right block, we first let \(\sigma\) be the permutation on \(n\) elements defined by \(\sigma(1) = 1\), and \(\sigma(j) = n + 1 - j\) for \(j \ne 1\). Then, for qubit \((\sigma(i), \sigma(j))\) with \(i \in [k]\) on the right block, we define \(f(\sigma(i), \sigma(j)) = v_j^{(i)}\).

Note that \(\sigma = \sigma^{-1}\). Then, for our classical cyclic code parity check matrices, \(P(\sigma) \partial P(\sigma) = \partial^\top \), where \(P(\sigma)\) is the permutation matrix associated with \(\sigma\).
This means that if \(s = \partial x\), \(P(\sigma) s = \partial^\top  P(\sigma)x\), and notably, if \(c \in \ker{\partial}\) then \(P(\sigma) c \in \ker{\partial^\top }\).
This will prove useful momentarily.
Furthermore, note that the first \(k\) bits form an information set for \(\partial\), and thus \(\{\sigma(i) : i \in [k]\}\) forms an information set for \(\partial^\top \).
We will take our canonical basis of logical operators in the right block to be supported on the rows and columns \(\{\sigma(i) : i \in [k]\}\).
The reason for taking this information set in the right block, as opposed to the first \(k\) rows and columns, will be explained in the following section, as it relates to how the \(X\) and \(Z\) basis extractors are connected.

The maximum degree of each kind of qubit in a single-basis system is listed in Table~\ref{tab:single-basis}, using the same notation as for Table~\ref{tab:single-column}.
Since we are now specializing to the case of cyclic HGP, we have \(w = r\).
We now describe the deformation map \(D_Z\).
Consider base code \(X\) check \(S\) with coordinate \((i, j)\) in the HGP code.
This corresponds to the operator \(X(\partial^\top  e_i \cdot e_j^\top , e_i \cdot (\partial e_j)^\top )\).
There are four cases:
\begin{enumerate}[label=Case \arabic*:, leftmargin=*]
    \item \(j \notin [k], \sigma(i)\notin [k]\). Then, this check has no overlap with the canonical \(Z\) logical operators and does not need to be deformed, \emph{i.e.} \(D_Z(S) = \emptyset\) and has degree \(2w\).
    \item \(j \in [k], \sigma(i) \notin[k]\). In this case, this check has no overlap with the canonical \(Z\) logical operators in the right block. For such a check, we assign \(D_Z(S)\) to the edges in \(G^{(j)}\) corresponding to \(D(S)\), \emph{i.e.} the single-column deformation map.
    The maximum degree of such a check qubit is still \(2w + u\).
    \item \(j \notin [k], \sigma(i) \in[k]\). In this case, this check has no overlap with the canonical \(Z\) logical operators in the left block. For such a check, we assign \(D_Z(S)\) to the edges in \(G^{(\sigma(i))}\) corresponding to \(\sigma(D(S)) = D(\sigma(S))\).
    The maximum degree of such a check qubit is \(2w + u\).
    \item \(j \in [k], \sigma(i) \in [k]\). We assign \(D_Z(S)\) to the union of \(D(S)\) in \(G^{(j)}\) and \(\sigma(D(S))\) in \(G^{(\sigma(i))}\).
    The maximum degree of such a check is \(2w + 2u\).
\end{enumerate}

\begin{table}[h]
    \centering
    \begin{tabular}{ccc}
    \toprule
    Qubit Type & Max. Degree & \makecell{Example Value for \\ \(d=10\) System}\\
    \midrule
    Left Block Data Qubit  & \(2w + 1\) & \(7\) \\
    Right Block Data Qubit & \(2w + 1\) & \(7\) \\
    Base Code \(X\) Check & \(2w + 2u\) & \(10\) \\
    Base Code \(Z\) Check & \(2w\) & \(6\) \\
    Ancilla System Data Qubit & \(2 + \rho + 2p + 1\) & 8 \\
    Vertex Check & \(\Delta_Z + 4\) & 8 \\
    Cycle Check & \(\max(\Delta_X, 4)\) & 6 \\
    Bridge Qubits & \(2 + \Delta_Z\) & 6\\
    \bottomrule
    \end{tabular}
    \caption{Qubit degrees in single-basis extractor system. The last column lists the values for the extractor system for the \([[882,50,10]]\) cyclic HGP code.
    }
    \label{tab:single-basis}
\end{table}

With the extractor \((\mathcal{X}_Z, f_Z, D_Z)\) suitably defined, we now show that, given that the single-column extractor \((\mathcal{X}_C, f, D)\) is distance preserving, then \((\mathcal{X}_Z, f_Z, D_Z)\) is as well. By Theorem \ref{thm:css-fault-dist}, this is sufficient to show that every measurement protocol performed using \((\mathcal{X}_Z, f_Z, D_Z)\) has fault distance \(d\).

\begin{theorem}
\label{thm:single-basis-dist}
    Let \(\Q = \HGP(\partial, \partial)\) be a cyclic HGP code of an \([n, k, d]\) classical cyclic code.
    Let \((\mathcal{X}_Z, f_Z, D_Z)\) be a logical \(Z\) extractor as described above.
    Then, for every possible logical \(Z\) operator measurement using the canonical basis, 
    \begin{equation}
    \label{eq:single-basis-dist}
        d(\mathcal{L}_Z^* \mathcal{S}_Z \mathcal{S}_V') \ge d
    \end{equation}
\end{theorem}

\begin{proof}
    Let \(\Zlog = Z(L, R)\), be the logical operator being measured, with \(L = \sum_{i=1}^k c_i \cdot e_i^\top \) for \(c_i \in \ker(\partial)\), and \(R = \sum_{j=1}^k e_{\sigma(j)} \cdot \tilde{c}_j^\top \) for \(\tilde{c}_j \in \ker(\partial^\top )\). To show Equation \eqref{eq:single-basis-dist}, it suffices to show that \(\wt(\Zlog' \Lambda) \ge d\), for an arbitrary logical \(\Zlog'\) in a different homology class than \(\Zlog\), and an arbitrary \(\Lambda \in \mathcal{S}_V'\). Without loss of generality, let \(\Zlog' = Z(L' + U\partial, R' +\partial U)\) where \(U \in \mathbb{F}_2^{n \times n}\), \(L' = \sum_{i=1}^k x_i \cdot e_i^\top \) for \(x_i \in \ker(\partial)\), and \(R' = \sum_{j=1}^{k} e_{\sigma(j)} \cdot \tilde{x}_j^\top \) for \(\tilde{x}_j \in \ker(\partial^\top )\). The proof proceeds through an analysis of cases of different possible values of \(x_i\) and \(\tilde{x}_j\).
    
    First, note that \(\Zlog'\) being in a different homology class than \(\Zlog\) implies that there exists an \(i\) such that \(c_i \ne x_i\), or there exists a \(j\) such that \(\tilde{c}_j \ne \tilde{x}_j\). If there exists an \(i\) such that \(c_i = 0\) and \(x_i \ne 0\), then Corollary \ref{cor:clean} implies that \(\Zlog'\) has support of size at least \(d\) in the columns in the left block outside of those which have qubits adjacent to the vertex checks of the extractor, and thus \(\wt(\Zlog' \Lambda) \ge d\). Similar reasoning holds if there is a \(j\) such that \(\tilde{c}_j = 0\) and \(\tilde{x}_j \ne 0\). Thus, for the rest of the proof, we may assume that if \(c_i = 0\), then \(x_i = 0\), and if \(\tilde{c}_j = 0\), then \(\tilde{x}_j = 0\).

    Let \(\eta_i\) be the set of vertex checks in \(\Lambda\) in \(V^{(i)}\)
    and let \(S_{B,i}\) be the support of \(\Lambda\) on \(B^{(i)}\). We can then write \(\wt(\Zlog' \Lambda)\) as
    \begin{align}
        \wt(\Zlog' \Lambda) &=
        \underbrace{\sum_{i=1}^k |x_i + f_{c_i}^{-1}(\eta_i) + (U\partial)_{(\cdot, i)}|}_{\text{Left block information columns}}
        + \underbrace{\sum_{i=k+1}^{n} |(U\partial)_{(\cdot, i)}|}_{\substack{\text{Left block non}\\\text{-information columns}}}
        + \underbrace{\sum_{j=1}^k |\tilde{x}_j + f_{\tilde{c}_j}^{-1}(P(\sigma)\eta_j) + (\partial U)_{(\sigma(j), \cdot)}|}_{\text{Right block information rows}}\nonumber \\
        &+ \underbrace{\sum_{j=k+1}^n |(\partial U)_{(\sigma(j), \cdot)}|}_{\substack{\text{Right block non}\\\text{-information rows}}} + \underbrace{\sum_{i = 1}^k |H_Z^T \eta_i|}_{\substack{\text{Single column}\\\text{extractors}}} 
        + \underbrace{\sum_{i=1}^{k-1} |\eta_i + \eta_{i+1}|}_{\text{Bridges}}
    \end{align}
    Suppose that there exists an \(i'\) such that \(x_{i'} \ne 0\) and \(x_{i'} \ne c_{i'}\). Then, 
    \begin{align}
        \wt(\Zlog' \Lambda) &\ge |(x_{i'} + f_{c_{i'}}^{-1}(\eta_i)) + (U\partial)_{(\cdot, i')}| + \sum_{i=k+1}^n |(U\partial)_{(\cdot, i)}| + |H_Z^T \eta_{i'}| \\
        & \ge d
    \end{align}
    where the second inequality follows from the fact that \((\mathcal{X}_C, f)\) is a distance preserving single-column extractor.
    Now suppose there exists a \(j'\) such that \(\tilde{x}_{j'} \ne 0\) and \(\tilde{x}_{j'} \ne \tilde{c}_{j'}\). Then,
    \begin{align}
        \wt(\Zlog' \Lambda) &\ge 
        |\tilde{x}_{j'} + f_{\tilde{c}_{j'}}^{-1}(P(\sigma)\eta_{j'}) + (\partial U)_{(\sigma(j), \cdot)} | + \sum_{j = k + 1}^n |(\partial U)_{(\sigma(j), \cdot)}| + |H_Z^T \eta_{j'}| \\
        & = |P(\sigma) \tilde{x}_{j'} + f_{P(\sigma)\tilde{c}_{j'}}^{-1}(\eta_{j'}) + (U'\partial)_{(\cdot, j')}| + \sum_{j=k+1}^n |(U' \partial)_{(\cdot, j)}| + |H_Z^T \eta_{j'}| \\
        & \ge d
    \end{align}
    where \(U' = P(\sigma) U^T P(\sigma)\)
    and the last inequality follows from the fact that \((\mathcal{X}_C, f)\) is a distance preserving single-column extractor.

    We are now left with the case that for all \(i \in [k]\), either \(x_i = 0\) or \(x_i = c_i\), and for all \(j \in [k]\), either \(\tilde{x}_j = 0\) or \(\tilde{x}_j = \tilde{c}_j\). Since \(\Zlog'\) is in a different homology class than \(\Zlog\), we also must have either at least one \(i\) such that \(x_i = 0\) and \(c_i \ne 0\), or at least one \(j\) such that \(\tilde{x}_j = 0\) and \(\tilde{c}_j \ne 0\). We further subdivide this into four cases.
    
    \begin{enumerate}[label=Case \arabic*:, leftmargin=*]
        \item There exists an \(i\) such that \(x_{i} = 0 \ne c_i\) and an \(i'\) such that \(x_{i'} \ne 0\).
        \item There exists a \(j\) such that \(\tilde{x}_j = 0 \ne \tilde{c}_j\) and a \(j'\) such that \(\tilde{x}_{j'} \ne 0\).
        \item For all \(i\), \(x_i = c_i\) and for all \(j\), \(\tilde{x}_j = 0\).
        \item For all \(i\), \(x_i = 0\) and for all \(j\), \(\tilde{x}_j = \tilde{c}_j\).
    \end{enumerate}

    We will prove cases one and three, as cases two and four follow from similar arguments.
    First, case one.
    Let \(i, i' \in [k]\) such that \(c_i \ne 0\), \(x_i = 0\), \(c_{i'} = x_{i'}\). For the moment suppose \(|i - i'| = 1\).
    Without loss of generality, let \(i = 1\) and \(i' = 2\).
    Additionally, we temporarily assume \(U = 0\).
    Let \(\eta_1 \in \mathbb{F}_2^{|V|}\) denote the set of vertex checks in \(G^{(1)}\) in \(\Lambda\), and \(\eta_2\) denote the set of vertex checks in \(G^{(2)}\) in \(\Lambda\).
    Let 
    \begin{equation}
    \beta = (f_{c_1}^{-1}(\eta_1) \cap c_2) \setminus f_{c_2}^{-1}(\eta_2)
    \end{equation}
    so that
    \begin{equation}
        |\beta| = |f_{c_1}^{-1}(\eta_1) \cap c_2| - |f_{c_1}^{-1}(\eta_1) \cap f_{c_2}^{-1}(\eta_2)|.
    \end{equation}
    Observe that
    \begin{equation}
    \label{eq:sdproof1}
        \wt(\Zlog' \Lambda) \ge |f_{c_1}^{-1}(\eta_1)| + |c_2 + f_{c_2}^{-1}(\eta_2)| + |H_Z^\top \eta_1| + |H_Z^\top \eta_2| + |\eta_1 + \eta_2|
    \end{equation}
    where \(H_Z^\top \eta_1 \in \mathbb{F}_2^{|E|}\) corresponds to the support of \(\Lambda\) on the edge qubits of \(G^{(1)}\) (and similarly for \(\eta_2\)). Suppose we were to use our single-column extractor to measure the operator \(Z((c_1 + c_2)\cdot e_1^\top , 0)\). The fact that our single-column extractor is distance preserving means that
    \begin{equation}
    \label{eq:sdproof2}
        |c_2 + f_{c_1 + c_2}^{-1}(\eta_1 \cup \eta_2)| + |H_Z^\top (\eta_1 \cup \eta_2)| \ge d
    \end{equation}
    where we think of \(\eta_1\cup \eta_2\) as the vertex checks we are using to try and reduce the weight of the logical \(Z(c_2\cdot e_1^\top , 0)\). Observe that
    \begin{equation}
    \label{eq:sdproof3}
        f_{c_1 + c_2}^{-1}(\eta_1 \cup \eta_2) = f_{c_1}^{-1}(\eta_1 \cup \eta_2) + f_{c_2}^{-1}(\eta_1 \cup \eta_2)
    \end{equation}
    and furthermore that
    \begin{align}
    \label{eq:sdproof4}
        |f_{c_2}^{-1}(\eta_1 \cup \eta_2) | &= |f_{c_2}^{-1}(\eta_1) \cup f_{c_2}^{-1}(\eta_2)|\nonumber \\
        &= | (c_2 \cap f^{-1}(\eta_1)) \cup f_{c_2}^{-1}(\eta_2)| \nonumber\\
        & \ge |(c_2 \cap f_{c_1}^{-1}(\eta_1)) \cup f_{c_2}^{-1}(\eta_2)| \nonumber\\
        &= |c_2 \cap f_{c_1}^{-1}(\eta_1)| + |f_{c_2}^{-1}(\eta_1)| - |c_2 \cap f_{c_1}^{-1}(\eta_1) \cap f_{c_2}^{-1}(\eta_2)|\nonumber\\
        &= |c_2 \cap f_{c_1}^{-1}(\eta_1)| + |f_{c_2}^{-1}(\eta_2)| - |f_{c_1}^{-1}(\eta_1) \cap f_{c_2}^{-1}(\eta_2)|\nonumber \\
        &= |\beta| + |f_{c_2}^{-1}(\eta_2)|
    \end{align}
    and that
    \begin{align}
    \label{eq:sdproof5}
        |f_{c_1}^{-1}(\eta_1 \cup \eta_2)| &= |f_{c_1}^{-1}(\eta_1) \cup f_{c_1}^{-1}(\eta_2 \setminus \eta_1)| \nonumber \\
        &\le | f_{c_1}^{-1}(\eta_1)| + |\eta_2 \setminus \eta_1| \nonumber \\
        &\le | f_{c_1}^{-1}(\eta_1)| + |\eta_1 + \eta_2|.
    \end{align}
    Putting this all together, we get that
    \begin{align}
        d + |\beta| &\le |c_2 + f^{-1}_{c_1 + c_2}(\eta_1 \cup \eta_2)| + |H_Z^\top (\eta_1 \cup \eta_2)| + |\beta| && \text{by equation \eqref{eq:sdproof2}} \\
        &= |c_2 + f^{-1}_{c_1}(\eta_1 \cup \eta_2) + f^{-1}_{c_2}(\eta_1 \cup \eta_2)| + |H_Z^\top (\eta_1 \cup \eta_2)| + |\beta| && \text{by equation \eqref{eq:sdproof3}} \\
        &\le |c_2| - |f_{c_2}^{-1}(\eta_1 \cup \eta_2)| + |f_{c_1}^{-1}(\eta_1 \cup \eta_2)| + |H_Z^\top  \eta_1| + |H_Z^\top \eta_2| + |\beta| && \text{by the triangle inequality} \\
        &\le |c_2| - |f_{c_2}^{-1}(\eta_2)| - |\beta| + |f_{c_1}^{-1}(\eta_1 \cup \eta_2)| + |H_Z^\top  \eta_1| + |H_Z^\top \eta_2| + |\beta| && \text{by equation \eqref{eq:sdproof4}} \\
        &\le |c_2| - |f_{c_2}^{-1}(\eta_2)| + |f_{c_1}^{-1}(\eta_1)| + |\eta_1 + \eta_2|  + |H_Z^\top  \eta_1| + |H_Z^\top  \eta_2| && \text{by equation \eqref{eq:sdproof5}} \\
        & \le \wt(\Zlog' \Lambda) && \text{by equation \eqref{eq:sdproof1}}.
    \end{align}

    Now, let us consider the ways in which including a nonzero \(U\) could reduce the weight of the checks.
    We consider only the support of \(\Zlog'\Lambda\) in the first two columns, and columns \( k < j \le n\).
    As has been established previously, if a row of \(U\partial\) has a nonzero entry in the first \(k\) indices, it must also have a nonzero entry in the last \(n-k\) indices.
    Thus, for each \(i \in f_{c_1}^{-1}(\eta_1) \cap(c_2 + f_{c_2}^{-1}(\eta_2))\), we cannot rule out the possibility that multiplication by stabilizers reduces the weight by one. Since \(f_{c_2}^{-1}(\eta_2) \subseteq c_2\), we may rewrite this set of indices as \(f_{c_1}^{-1}(\eta_1) \cap(c_2 + f_{c_2}^{-1}(\eta_2)) = f_{c_1}^{-1}(\eta_1) \cap (c_2 \setminus f_{c_2}^{-1}(\eta_2)) = \beta\). Thus, in these \(n-k + 2\) columns of the left block, multiplication by stabilizers can reduce the weight by at most \(|\beta|\), which implies that \(\wt(\Zlog' \Lambda) \ge d\) as desired.

    The general case where \(|i - i'| \ne 1\) follows from observing that
    \begin{align}
        \wt(\Zlog' \Lambda) &\ge |f_{c_i}^{-1}(\eta_i)| + |c_{i'} + f_{c_{i'}}^{-1}(\eta_{i'})| + |H_Z^\top \eta_{i}| + |H_Z^\top \eta_{i'}| + \sum_{l=i}^{i'-1} |\eta_{l} + \eta_{l+1}| \\
        & \ge |f_{c_i}^{-1}(\eta_i)| + |c_{i'} + f_{c_{i'}}^{-1}(\eta_{i'})| + |H_Z^\top \eta_{i}| + |H_Z^\top \eta_{i'}| + |\eta_{i} + \eta_{i'}|
    \end{align}
    where we assume without loss of generality that \(i' > i\), and the second inequality follows from the triangle inequality. The rest of the argument goes through identically.

    Now for case three.
    Suppose for all \(i \in [k]\), \(x_i = c_i\), and that for all \(j \in [k]\), \(\tilde{x}_j = 0\).
    Note that at least one \(x_{i'} \ne 0\), otherwise \(\Zlog'\) is in the stabilizer.
    Additionally, at least one \(\tilde{c}_{j'} \ne 0\), as otherwise \(\Zlog'\) is in the same homology class as \(\Zlog\).
    Suppose again temporarily that \(U = 0\), and also that \(i' \ne j'\).
    We then see that
    \begin{align}
        \wt(\Zlog' \Lambda) &\ge |c_{i'} + f_{c_{i'}}^{-1}(\eta_{i'})| + |f_{\tilde{c}_{j'}}^{-1}(P(\sigma)\eta_{j'})| + |c_{j'} + f_{c_{j'}}^{-1}(\eta_{j'}))| + |H_Z^T \eta_{i'}| + |H_Z^T \eta_{j'}| + \sum_{i=1}^{k-1} | \eta_{i} + \eta_{i+1}| \\
        &\ge |c_{i'} + f_{c_{i'}}^{-1}(\eta_{i'})| + |c_{j'} + f_{c_{j'} + P(\sigma)\tilde{c}_{j'}}^{-1}(\eta_{j'})| + |H_Z^{\top} \eta_{i'}| + |H_Z^{\top}\eta_{j'}| + \sum_{i=1}^{k-1}|\eta_i + \eta_{i+1}|\\
        &= \wt(\Zlog'' \Lambda)
    \end{align}
    where the second inequality follows from the triangle inequality and Equation~\eqref{eq:sdproof3}, \(\Zlog'' = Z(c_{i'}\cdot e_{i'}^{\top} + c_{j'}\cdot e_{j'}^{\top}, 0)\) and in the last equality, we assume we are measuring the operator
    \begin{equation}
        Z(c_{i'}\cdot e_{i'}^{\top} + (c_{j'} + P(\sigma)\tilde{c}_{j'}) \cdot e_{j'}^{\top}, 0).
    \end{equation}
    If \(c_{j'} = 0\), then \(\wt(\Zlog'' \Lambda ) \ge d\) since this corresponds to Case 1 above.
    If \(c_{j'} \ne 0\), then
    \begin{align}
        \wt(\Zlog''\Lambda) &\ge |c_{j'} + f^{-1}_{c_{j'} + P(\sigma) \tilde{c}_{j'}}(\eta_{j'})| + |H_Z^{\top} \eta_{j'}| \\
        & \ge d
    \end{align}
    where the second inequality follows because the single-column system is distance preserving.
    
    Now suppose \(i' = j'\). 
    \begin{align}
        \wt(\Zlog'\Lambda) &\ge |c_{i'} + f_{c_{i'}}^{-1}(\eta_{i'})| + |f_{P(\sigma)\tilde{c}_{i'}}^{-1}(\eta_{i'})| + |H_Z^\top \eta_{i'}| \\
        &\ge |c_{i'} + f_{c_{i'} + P(\sigma)\tilde{c}_{i'}}^{-1}(\eta_{i'})| + |H_Z^\top  \eta_{i'}| \\
        &\ge d
    \end{align}
    where the second inequality follows from the triangle inequality and equation \eqref{eq:sdproof3}, and the third inequality follows because the single-column system is distance preserving.

    Finally, by similar reasoning as above, multiplication by \(Z\) checks cannot reduce this weight below \(d\).
\end{proof}

We conclude this section by noting that while we have restricted the above discussion to cyclic HGP codes, the same techniques would apply to HGP codes \(\HGP(\partial, \partial)\) where \(\partial\) is full rank, as in this case, there are no logical operators supported on the right block.

\subsection{Full Extractors}
\label{sec:full-ext}

Now, by using one \(X\) extractor, and one \(Z\) extractor as described in the previous section, it is possible to perform any logical CSS measurement.
In this section, we demonstrate how to connect these two extractors to perform non-CSS measurements, thus yielding a full extractor.

Let \((\mathcal{X}_X, f_X, D_X)\) be the \(X\) extractor and \((\mathcal{X}_Z, f_Z, D_Z)\) be the \(Z\) extractor for a cyclic HGP code, using an \([n, k, d]\) cyclic code as the classical code.
Consider that the domains of \(f_X\) and \(f_Z\) overlap in \(2k^2\) locations, corresponding to the \(k^2\) information qubits in the left block, and the \(k^2\) information qubits in the right block.
Let the two extractor graphs be \(G_X\) and \(G_Z\), and let \emph{e.g.} vertex \(v_{i, x}^{(j)}\) denote the \(i\)th vertex in the \(j\)th single-column extractor graph in \(G_X\).
Suppose without loss of generality that for each single-column extractor graph, vertices \(1, 2, \dots, k\) are the ones which the port function maps the information qubits to.
We see that the \(2k^2\) qubits on which the domains of \(f_X\) and \(f_Z\) overlap correspond to only \(k^2\) vertices in \(G_X\) and \(k^2\) vertices in \(G_Z\), which is why we chose the information rows and columns in the right block to correspond to \(\{\sigma(i): i \in [k]\}\), as opposed to just \([k]\).
To connect the two extractors, we construct the full extractor graph \(G_F\) by adding an edge between vertices \(v_{i, x}^{(j)}\) and \(v_{j, z}^{(i)}\) for every \(i, j \in [k]\), as each of these two vertices has the same pre-image under \(f_X\) and \(f_Z\) respectively. This adds \(k^2\) edges to the graph, and \(k^2 - 1\) additional cycles. The following lemma characterizes the weight of these new cycles, as well as the new congestion introduced by adding these cycles.

\begin{lemma}
	Let \(G_X\) and \(G_Z\) be the \(X\) and \(Z\) extractor graphs as described in the previous section. Suppose they each have a cycle basis with maximum weight \(\tilde{w}\) and congestion \(\tilde{\rho}\). Let \(T\) be a \((k-1) \times |E|\) matrix such that row \(i\) of \(T\) corresponds to a path in the single-column extractor graph between vertex \(i\) and \(i+1\). Let the maximum row weight of \(T\) be \(w'\) and the maximum column weight of \(T\) be \(\rho'\). Then, there exists a cycle basis of the full extractor graph with maximum cycle weight \(\max(\tilde{w}, w' + 3)\) and congestion at most \(\tilde{\rho} + \rho'\).
\end{lemma}

\begin{proof}
    We will explicitly construct \(k^2-1\) new linearly independent cycles.
    Our construction borrows heavily from the universal adaptors of~\cite{swaroop2024universal}.
    Let \(e_{i,j}\) be the edge between \(v_{j,x}^{(i)}\) and \(v_{i,z}^{(j)}\).
    Let \(b_{i,x}^{(j)}\) be the bridge qubit/edge between \(v_{i,x}^{(j)}\) and \(v_{i,x}^{(j+1)}\), and \(b_{i,z}^{(j)}\) be the bridge qubit/edge between \(v_{i,z}^{(j)}\) and \(v_{i,z}^{(j+1)}\).
    Consider the following set of cycles:
    \begin{enumerate}
        \item For \(i \in [k-1]\), \(j \in [k]\), the cycle \(\{b_{j,x}^{(i)}, e_{i,j}, e_{i+1,j}, t_{i, z}^{(j)}\} \) where \(t_{i,z}^{(j)}\) is the set of edges in the \(j\)th \(Z\) column extractor corresponding to row \(i\) of \(T\).
        \item For \(i \in [k-1]\), \(i\) odd, the cycle \(\{b_{k,z}^{(i)}, e_{k,i}, e_{k,i+1}, t_{i,x}^{(k)} \}\) where \(t_{i,x}^{(k)}\) is the set of edges in the \(k\)th \(X\) row extractor corresponding to row \(i\) of \(T\).
        \item For \(i \in [k-1]\), \(i\) even, the cycle \(\{b_{1,z}^{(i)}, e_{1,i}, e_{1,i+1}, t_{i,x}^{(1)}\}\) where \(t_{i,x}^{(1)}\) is the set of edges in the first \(X\) row extractor corresponding to row \(i\) of \(T\).
    \end{enumerate}
    The bridge edges present in this set of cycles are shown in Figure~\ref{fig:full-ext-cycle-basis-bridge-edges}.
    Note how the bridge edges together form a path through the information vertices, and we can see that the two edges contained at the ends of this path will each be contained in one new cycle, and every other newly added edge will be contained in two cycles.
    \begin{figure}
        \centering
        \includegraphics[width=0.7\linewidth]{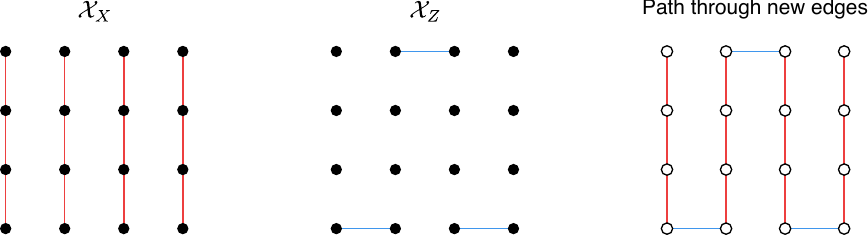}
        \caption{
        Three different subsets of the full extractor, here shown for \(k=4\).
        On the left, black dots indicate information vertices in \(\mathcal{X}_X\), and red lines indicate edges in \(\{B_X^{(1)}, \dots, B_X^{(k)}\}\) which are contained in the newly added cycles for the cycle basis of the full extractor.
        In the middle, black dots indicate information vertices in \(\mathcal{X}_Z\), and blue edges indicate edges in \(\{B_Z^{(1)}, \dots, B_Z^{(k)}\}\) which are contained in the newly added cycles for the cycle basis of the full extractor.
        Each red and blue edge is contained in one new cycle, which consists of the aforementioned edge, two newly added bridge edges corresponding to the vertices incident to the edge, and then a path in the opposite basis graph corresponding to a row of \(T\).
        On the right, white dots indicate bridge edges between \(\mathcal{X}_X\) and \(\mathcal{X}_Z\), and we overlay the red and blue edges to show they form a path through the new bridge edges.
        This ensures that the cycle basis matrix, when restricted to these newly added edges, looks like the full rank parity check matrix for a repetition code on \(k^2\) bits.}
        \label{fig:full-ext-cycle-basis-bridge-edges}
    \end{figure}
    Let \(N_X\) and \(N_Z\) be the cycle bases of \(G_X\) and \(G_Z\).
    Then, we may write the cycle basis for \(G_F\) as

    \begin{equation}
        \begin{pmatrix}
            N_X & 0 & 0 \\
            T_X & H_C & T_Z \\
            0 & 0 & N_Z
        \end{pmatrix}
    \end{equation}

    where \(H_C\) is the standard full rank parity check matrix for the repetition code on \(k^2\) bits, and the middle \(k^2 - 1\) rows correspond to the cycles listed above.
    Since \(N_X, H_C\), and \(N_Z\) all have linearly independent rows, the above matrix has linearly independent rows, and thus this forms a full cycle basis for \(G_F\).
    We can see that the maximum row weight of \(\begin{pmatrix}
        T_X & H_C & T_Z
    \end{pmatrix}\) is \(w' + 3\), and the maximum column weight is \(\max(2, \rho')\).
    Thus, the maximum weight of the cycle basis is \(\max(\tilde{w}, w' + 3)\), and the congestion is at most \(\tilde{\rho} + \rho'\).
\end{proof}

We note that in practice, we found it straightforward to construct single-column extractor graphs such that \(w'\) and \(\rho'\) were small (e.g. see Table~\ref{tab:full-ext}). 

We now define the port function of the full extractor. As stated previously, the intersection of the domain of \(f_X\) and the domain of \(f_Z\) is the information qubits in the code. Thus, to construct a valid port function, we must assign each information qubit to its corresponding vertex in either the \(X\) extractor, or the \(Z\) extractor. A simple choice is to assign all of the information qubits to the \(X\) extractor, e.g. for each qubit \(q\) in the support of a canonical logical basis operator,
\begin{align}
    \label{eq:full-ext-port}
    f_F(q) = \begin{cases} 
    f_X(q) & \text{ if } q \in \mathrm{domain}(f_X) \\
     f_Z(q) & \text { otherwise }
\end{cases}
\end{align}

In practice, it might be beneficial to include connections between the information qubits and the corresponding vertex checks of both the \(X\) and \(Z\) extractors, as this allows for measuring disjoint \(X\) and \(Z\) logical operators in parallel.

Next, we examine the necessary adjustments to the deformation map that arise from using the port function from Equation~\eqref{eq:full-ext-port}.
For a \(Z\) check \(S\), we simply take \(D_F(S) = D_X(S)\).
Similarly, for an \(X\) check \(S\) with no support on any of the information qubits in the code, we take \(D_F(S) = D_Z(S)\).
For an \(X\) check which has support on an information qubit, it becomes necessary to add a connection between that \(X\) check and the bridge qubit adjacent to the \(Z\) extractor vertex check in the image of that information qubit under \(f_Z\), as illustrated in Figure \ref{fig:extractor-connectivity}.
We note that for a cyclic code parity check matrix, the maximum overlap between a given check and the first \(k\) columns is \(w-1\) (where \(w\) is the weight of the rows).
This follows from the fact that given the check polynomial \(h(x) = h_0 + \dots + h_k x^k\), we always have \(h_0 = h_k = 1\).
Thus for such \(X\) checks, \(D_F(S) = D_Z(S) \cup B_S\), where \(B_S\) are the aforementioned bridge qubits.
The degree of each kind of qubit in the construction is listed in Table~\ref{tab:full-ext}.

\begin{figure}
    \centering
    \includegraphics[width=0.6\linewidth]{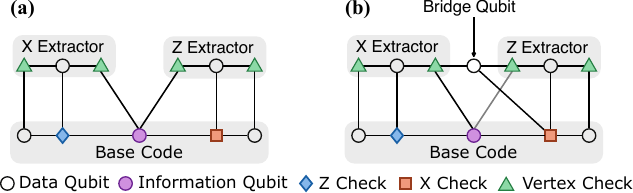}
    \caption{\textbf{(a)} The required connectivity between the \(X\) extractor and the base code, as well as \(Z\) extractor and the base code. For simplicity of the illustration, for each extractor, we include one edge incident to a vertex in the image under the port function of an information qubit. \textbf{(b)} The required connectivity between the full extractor and the base code. Observe that in order for the base code \(X\) check to be connected to a path matching between the vertex checks in the extractor that it locally anti-commutes with on the base code, an additional connection to the bridge qubit must be added. We also note that while technically not part of the full extractor, we keep the connection between the information qubit and the \(Z\) extractor vertex check to enable parallel measurement of disjoint \(X\) and \(Z\) operators.}
    \label{fig:extractor-connectivity}
\end{figure}

\begin{table}[h]
    \centering
    \begin{tabular}{ccc}
        \toprule
         Qubit Type & Max. Degree & \makecell{Example Value for \\ \(d=10\) System}  \\
         \midrule
         Left Block Data Qubit & \(2w + 2\) & 8 \\
         Right Block Data Qubit & \(2w + 2\) & 8 \\
         Base Code \(X\) Check & \(2w + 2u + 2(w-1)\) & 14 \\
         Base Code \(Z\) Check & \(2w + 2u\) & 10 \\
         Ancilla System Data Qubit & \(2 + \rho + 2p + 1 + \rho'\) & 8 \\
         Vertex Check & \(\Delta_Z + 5\) & 9 \\
         Cycle Check & \(\max(\Delta_X, w' + 3 )\) & 6 \\
         Single Basis Bridge Qubit & \(\Delta_Z + 3\) & 7 \\
         Full Extractor Bridge Qubit & \(4 + 2w\) & 10\\
         \bottomrule
    \end{tabular}
    \caption{Qubit degrees in a full extractor system. The last column lists the values for the extractor system for the \([[882,50,10]]\) cyclic HGP code. Note that some \(X\) checks in the \(d=10\) system have degree greater than ten. To reduce the degree of the construction back to ten, measure these checks using Bell checks~\cite{yoder2025tour}, which use two degree-8 qubits instead of one degree-\(14\) qubit.}
    \label{tab:full-ext}
\end{table}

Now, having suitably defined the full extractor, we show its fault tolerance.

\begin{theorem}
    \label{thm:full-ext}
    Let \((\mathcal{X}_F, f_F)\) be a full extractor described above for a cyclic HGP code with a classical \([n, k, d]\) code. Suppose the single-column extractor is distance preserving, and the single-column extractor graph \(G\) has edge connectivity \(\lambda(G)\). Then, using this extractor to measure any product of canonical logical basis operators supported on the code has fault distance at least \(\min(d, k^2, \lambda(G) \cdot k)\).
\end{theorem}

\begin{proof}
    Suppose we use our extractor to measure logical operator \(\overline{P}_M =  \Xlog \Zlog\) (where we ignore phase). Because our merged code may potentially be non-CSS, we can no longer rely on Theorem \ref{thm:css-fault-dist}, and must instead fall back to Lemma \ref{thm:fault-dist}. We first consider the spacelike fault distance, \emph{i.e.} we wish to evaluate the smallest weight operator in the set \(\mathcal{L}^*\mathcal{G}\). Any operator in this set may be written (up to phase) as
    \begin{equation}
    \label{eq:full-ft-1}
        \Xlog' \Zlog' X_e M M'
    \end{equation}
    where \(\Xlog'\) and \(\Zlog'\) are \(X\) and \(Z\) type logical operators of the base code such that \(\Xlog' \Zlog' \ne \overline{P}_M\), \(X_e\) is the product of \(X\) operators on edge qubits of the extractor, \(M \in \mathcal{S}\) and \(M' \in \mathcal{S}'\). Note that operators of this form actually include all elements of the set \(\overline{Q}_M \mathcal{L} \mathcal{G}\), and so this analysis extends to the timelike distance as well. Furthermore, every operator \(M M'\) can be expressed as the product of an element of \(\mathcal{S}\), \(X\) operators on edge qubits in the extractor, and an element \(\Lambda \in \mathcal{S}_V\). Thus, we may re-write Equation \eqref{eq:full-ft-1} as
    \begin{equation}
        \Xlog' \Zlog' X_e M \Lambda
    \end{equation}
    and we see that
    \begin{equation}
        \wt(\Xlog' \Zlog' X_e M \Lambda) \ge \wt(\Xlog' \Zlog' M \Lambda).
    \end{equation}

    Let \(\Zlog'\Xlog' M = Z(L_Z +U_Z \partial, R_Z + \partial U_Z) X(L_X + U_X\partial, R_X + \partial U_X)\) for some \(L_Z, R_Z, U_Z, L_X, R_X\) and \(U_X\).
    
    Now, suppose that \(I \ne \Zlog' \ne \Zlog\). We then have 
    \begin{align}
        \wt(\Zlog'\Xlog' M \Lambda) &\ge |L_Z + U_Z \partial + f^{-1}_{Z, L, \Zlog}(\eta_z) | + |R_Z +\partial U_Z  + f^{-1}_{Z, R, \Zlog}(\eta_z) | + |S_{A, Z}| \\ &\ge d
    \end{align}
    where \(\eta_z \in \mathbb{F}_2^{k \cdot |V|}\) denotes the set of vertices in \(G_Z\) which are part of \(\Lambda\), \(S_{A, Z}\) denotes the support of \(\Lambda\) on the edge qubits of the \(Z\) extractor, \(f_{Z, L, \Zlog}\) denotes the \(Z\) extractor port function with the domain restricted to the support of \(\Zlog\) in the left block, and similarly for \(f_{Z, R, \Zlog}\). The second inequality follows from the fact that the single-column extractor is distance preserving, and from Theorem \ref{thm:single-basis-dist}. Similarly, if \(I \ne \Xlog' \ne \Xlog\), we also have that \(\wt(\Zlog'\Xlog' M \Lambda) \ge  d\). Thus, it remains to examine the case where \(\Zlog' = \Zlog\) and \(\Xlog' = I\) (and vice versa). Suppose
    \begin{equation}
        L_Z = \sum_{i=1}^k c_{i,z} \cdot e_i^\top , \quad R_Z = \sum_{j=1}^k e_{\sigma(j)} \cdot \tilde{c}_{j,z}^\top , \quad L_X = \sum_{j=1}^k e_j \cdot c_{j,x}^\top , \quad R_X = \sum_{i=1}^k \tilde{c}_{i,x} \cdot e_{\sigma(i)}^\top .
    \end{equation}
    Furthermore, suppose that in the \(i\)th single-column \(Z\) extractor, \(\Lambda\) corresponds to vertices \(\eta_{i,z}\), and in the \(i\)th row of the \(X\) extractor, \(\Lambda\) corresponds to vertices \(\eta_{i,x}\). The proof now proceeds in three steps. First, we show that if any of the \(\eta_{i,z} = 0\), then \(\wt(\Zlog M\Lambda) \ge d \). Then, we show that if any of the \(\eta_{i,x} = \mathbf{1}\) (the all ones vector), then \(\wt(\Zlog M \Lambda) \ge d\). Then, assuming each \(\eta_{i,z} \ne 0\) and each \(\eta_{i,x} \ne \mathbf{1}\), we show \(\wt(\Zlog M \Lambda) \ge \min(k^2, k\cdot \lambda(G))\).

    Suppose there exists at least one \(i' \in [k]\) such that \(\eta_{i', z} = 0\). We then let \(i,j \in [k]\) be the indices that minimize \(|i-j|\) such that \(\eta_{i,z} = 0\), and \(\eta_{j,z} \ne 0\) and \(c_{j,z} \ne 0\) or \(\tilde{c}_{j, z} \ne 0\). We know such a \(j\) must exist, as there must be at least one index \(j\) such that \(c_{j,z} \ne 0\) or \(\tilde{c}_{j, z} \ne 0\) (as otherwise \(\Zlog = I\)) and for that \(j\), if \(\eta_{j,z} = 0\) then we automatically have \(\wt(\Zlog M\Lambda) \ge d\). Without loss of generality, let \(j > i\), and let \(c_{j,z} \ne 0\). Note that for any \(i < l < j\), we must have that \(\eta_{l, z} \ne 0\). Furthermore, we have that \(c_{i,z} = 0\) and \(\tilde{c}_{i, z} = 0\), as otherwise we again automatically have \(\wt(\Zlog M \Lambda)\ge d\). Now, supposing for the moment that \(M = I\), we have
    \begin{align}
        \wt(\Zlog \Lambda) &\ge |\eta_{i+1, z}| + \sum_{l=i+1}^{j-1} |\eta_{l,z} + \eta_{l+1,z}| + |c_{j,z} + f^{-1}_{c_{j,z}}(\eta_{j,z})| \\
        & \ge |\eta_{i+1, z}| + \sum_{l=i+1}^{j-1}\max(|\eta_{l,z}| - |\eta_{l+1,z}|, |\eta_{l+1,z}| - |\eta_{l,z}|) + |c_{j,z}| - |f_{c_{j,z}}^{-1}(\eta_{j,z})|.
    \end{align}

    If we split out the last term from the sum, we get
    \begin{align}
        \wt(\Zlog \Lambda) & \ge |\eta_{i+1, z}| + \left (\sum_{l=i+1}^{j-2}\max(|\eta_{l,z}| - |\eta_{l+1,z}|, |\eta_{l+1,z}| - |\eta_{l,z}|) \right )\nonumber \\ &+ \max(|\eta_{j-1,z}| - |\eta_{j,z}|, |\eta_{j,z}| - |\eta_{j-1,z}|) + |c_{j,z}| - |f_{c_{j,z}}^{-1}(\eta_{j,z})| \\
        & \ge |\eta_{i+1, z}| + \left (\sum_{l=i+1}^{j-2}\max(|\eta_{l,z}| - |\eta_{l+1,z}|, |\eta_{l+1,z}| - |\eta_{l,z}|) \right ) + |c_{j,z}| - \min(|\eta_{j-1,z}|, |\eta_{j,z}|)
    \end{align}
    where we used that \(|\eta_{j,z}| \ge |f_{c_{j,z}}^{-1}(\eta_{j,z})|\). We may continue this process, until we obtain
    \begin{align}
        \wt(\Zlog \Lambda) & \ge |\eta_{i+1, z}| + |c_{j,z}| - \min(|\eta_{i+1,z}|, |\eta_{i+2,z}|) \\
        & \ge |c_{j,z}|\\
        & \ge d.
    \end{align}

    Since we are only relying on the support of our operator in a single information row or column of the base code, this weight cannot be reduced by multiplication by stabilizers by Lemma~\ref{lem:hgp-clean}.

    Now we move to the second step.
    Assume there exists at least one \(i\) such that \(\eta_{i,x} = \mathbf{1}\).
    Again, note that there must exist at least one \(j\) such that \(c_{j,x} \ne 0\) or \(\tilde{c}_{j,x} \ne 0\), as otherwise the operator we are considering cleaning is in the same homology class as the operator we are measuring.
    Let \(i, j \in [k]\) be the indices that minimize \(|i-j|\), such that \(\eta_{j,x} = \mathbf{1}\) and either \(c_{i,x} \ne 0\) or \(\tilde{c}_{i,x} \ne 0\).
    Without loss of generality let \(j > i\) and let \(c_{i,x} \ne 0\).
    For each \(i'\), let \(|\eta_{i',x}| = n - |c_{i,x}| + \varepsilon_{i'}\).
    Again, supposing temporarily that \(M = I\), we have
    \begin{align}
        \wt(\Zlog \Lambda) &\ge |f_{c_{i,x}}^{-1}(\eta_{i,x})| + \sum_{l=i}^{j-2} |\eta_{l,x} + \eta_{l+1,x}| + |\eta_{j-1,x} + \eta_{j,x}| \\
        & \ge |f_{c_{i,x}}^{-1}(\eta_{i,x})| + \sum_{l=i}^{j-2}|\varepsilon_{l} - \varepsilon_{l+1}| + n - |\eta_{j-1,x}| \\
        &\ge |f_{c_{i,x}}^{-1}(\eta_{i,x})| + \sum_{l=i}^{j-2} \max( \varepsilon_l - \varepsilon_{l+1}, \varepsilon_{l+1}-\varepsilon_l) + |c_{i,x}| - \varepsilon_{j-1} \\
        &\ge |f_{c_{i,x}}^{-1}(\eta_{i,x})| + |c_{i,x}| - \min(\varepsilon_i, \varepsilon_{i+1}) \\
        & \ge \max(0, \varepsilon_i) + |c_{i,x}| - \min(\varepsilon_i, \varepsilon_{i+1}) \\
        & \ge d
    \end{align}
    where we used the same telescoping sum as in the previous step. Again, since we are only relying on the support of our operator in single information row/column of the base code, this weight cannot be reduced by multiplication by stabilizers by Lemma~\ref{lem:hgp-clean}.
    
    This brings us to the third step. Suppose for every \(i\), \(\eta_{i,z} \ne 0\), and for every \(j\), \(\eta_{j,x} \ne \mathbf{1}\). 
    Additionally, suppose there are \(n_z\) indices such that \(\eta_{i,z} \ne \mathbf{1}\).
    For each such \(i\), the \(i\)th \(Z\) column extractor has support of size at least \(\lambda(G)\).
    Additionally, suppose there are \(n_x\) indices such that \(\eta_{j,x} \ne 0\).
    For each such \(j\), the \(j\)th \(X\) row extractor has support at least \(\lambda(G)\). 
    Looking only at the extractor, but taking into account the support of the bridge between the \(X\) and \(Z\) extractors, we then have
    \begin{align}
        \wt(\Zlog M \Lambda) &\ge n_z \lambda(G) + (k-n_z) \cdot k - n_x(k - n_z) + n_x \lambda(G) \\
        & = (n_z + n_x) \lambda(G) + (k-n_z) (k-n_x).
    \end{align}
    If \(n_x + n_z \ge k\), then we have \(\wt(\Zlog M \Lambda) \ge k \lambda(G)\) as desired.
    Now suppose \(n_x + n_z < k\). Then
    \begin{align}
        (n_z + n_x) \lambda(G) + (k-n_z)(k-n_x) & \ge (n_z + n_x) \lambda(G) + k(k-n_z-n_x) \\
        & \ge k \cdot \min(k, \lambda(G)).
    \end{align}

    This completes the proof.
\end{proof}
Finally, observe that the full extractor has size
\begin{equation}
    4k |E| + 4 (k-1) |V| + 2k^2 + 1 = \Theta(k(|E| + |V| + k)).
\end{equation}
This observation, along with Theorem~\ref{thm:full-ext}, proves Theorem~\ref{thm:main}. 

\subsection{Explicit Construction Details}

In this section we provide additional information about the three new extractors shown in Table~\ref{tab:comparison}; namely, the size of the graph, the edge connectivity of the graph, and the distribution of qubit degrees of the full extractor.
We note that each construction contains several qubits of degree greater than ten.
In each case, these qubits correspond to deformed \(X\) checks of the base code.
In order to bring the maximum qubit degree of the construction down to ten, we measure these checks using a Bell pair; see \emph{e.g.} Figure 4 of~\cite{yoder2025tour}.
Given a check of degree \(\delta\), this results in two qubits of degrees \(\lfloor \frac{\delta}{2} \rfloor + 1\) and \(\lceil \frac{\delta}{2} \rceil + 1\). The full degree distributions are listed in Table~\ref{tab:deg-dist}.

\begin{table}[h]
	\begin{center}
\begin{tabular}{cccccccccccc}
	\toprule
	& \multicolumn{3}{c}{Graph} & \multicolumn{8}{c}{Qubit Degree Distribution} \\
	\cmidrule(lr){2-4} \cmidrule(lr){5-12}
	\addlinespace
	Code & \(|V|\) & \(|E|\) & \(\lambda(G)\) & 3 & 4 & 5 & 6 & 7 & 8 & 9 & 10 \\
	\midrule
	\([[450,32,8]]\) & 15 & 30 & 3 & 4 & 201 & 86 & 664 & 294 & 265 & 52 & 39 \\
	\addlinespace
	\([[882,50,10]]\) & 21 & 42 & 4 & 0 & 402 & 89 & 1414 & 473 & 504 & 90 & 39 \\
	\addlinespace
	\([[1922,50,16]]\) & 31 & 62 & 4 & 0 & 557 & 117 & 3309 & 761 & 774 & 82 & 51\\
	\bottomrule
\end{tabular}
\caption{Graph sizes, graph edge connectivity, and qubit degree distributions for the three extractors proposed in this work.
The qubit degree distributions include all qubits in both the base code and extractor.}
\label{tab:deg-dist}
\end{center}
\end{table}

\section{Numerical Simulation Details}
\label{sec:numerics-details}
In this section we provide additional details about the memory and surgery experiments described in Section~\ref{sec:simulation-results}. All simulations use a standard circuit-level depolarizing noise model parameterized by a single value \(p\), where qubit initialization, measurement, single-qubit gates, and two-qubit gates have errors with probability \(p\), and idling locations have errors with probability \(p/10\). 

The memory experiment for the base HGP code proceeded with the following steps, all of which were noisy:
\begin{itemize}
    \item Initialize all data qubits in \(\ket{0}\).
    \item Perform the syndrome measurement circuit \(N_R\) times.
    \item Measure all of the data qubits in the \(Z\) basis.
\end{itemize}

For the surgery experiment, we used the full extractor to measure the following logical operator:
\begin{align}
    \overline{P}_1 &= X_0 X_1Y_3X_5Y_8Z_9Y_{10}Y_{12}Y_{13}Y_{14}Z_{16}Z_{18}Y_{19}Z_{21}Z_{22}Z_{23}X_{24}Z_{25}Y_{26}\nonumber\\
    & \cdot Z_{28}Z_{29}Z_{30}Z_{32}X_{33}Z_{34}Z_{35}X_{36}Y_{37}Y_{38}Y_{39}Z_{42}Z_{43}Y_{44}Y_{46}X_{47}Z_{48}Z_{49}.
\end{align}
We constructed a symplectic basis \(\{\overline{P}_1, \overline{Q}_1, \dots, \overline{P}_{50}, \overline{Q}_{50}\}\), and then the experiment proceeded as follows:
\begin{itemize}
    \item Noiselessly prepare the data code block in a joint \(+1\) eigenstate of \(\overline{P}_1, \dots, \overline{P}_{50}\).
    \item Initialize the data qubits of the extractor in \(\ket{+}\).
    \item Measure the syndrome of the merged code \(N_R\) times.
    \item Measure the data qubits of the extractor in the \(X\) basis.
    \item Measure the syndrome of the base code once.
    \item Noiselessly measure the syndrome of the base code, and noiselessly measure \(\overline{P}_2, \dots, \overline{P}_{50}\).
\end{itemize}
All circuits were constructed using the ``Chunked Circuit Construction'' approach described in Appendix B of~\cite{gidney2024magic}, and the syndrome measurement circuits were scheduled using the integer programming method described in Section A.5 of~\cite{yoder2025tour}.

For the first stage of the two-stage decoder, we use Relay-BP~\cite{muller2025improved} with the parameters listed in Table~\ref{tab:relay-params}.
For the second stage, we use an integer linear programming (ILP) decoder implemented with Gurobi~\cite{gurobi}.

In Figure~\ref{fig:extractor-results-detailed}, we show the results of the logical measurement error rate using just the Relay-BP decoder, and an estimate of the logical measurement error rate using an ILP decoder.
For the Relay-BP decoder, we note that we only count a shot as a logical measurement failure if, for a predicted error \(e'\) and actual error \(e\), the components of \(Le\) and \(Le'\) corresponding to the measured observable differ, where \(L\) is the logical observables matrix~\cite{higgott2025sparse}.
We do not require that Relay-BP converges, \emph{i.e.} \(De' = 0\), where \(D\) is the detector check matrix~\cite{higgott2025sparse}.
For the ILP decoder, we take all failed shots from Relay-BP (including shots in which Relay-BP converged to an incorrect prediction) and decode these shots with the ILP decoder.
We note that it is possible that there are shots that Relay-BP decoded correctly that the ILP decoder would decode incorrectly, thus the results shown in Figure~\ref{fig:extractor-results-detailed}(b) are a lower bound on the actual ILP decoder results, although we expect this to be tight in practice.
As seen in the figure, for \(N_R \ge 11\), there is a significant discrepancy in the logical measurement error rates of Relay-BP and the two-stage decoder.

\begin{figure}
    \centering
    \includegraphics[width=1.0\linewidth]{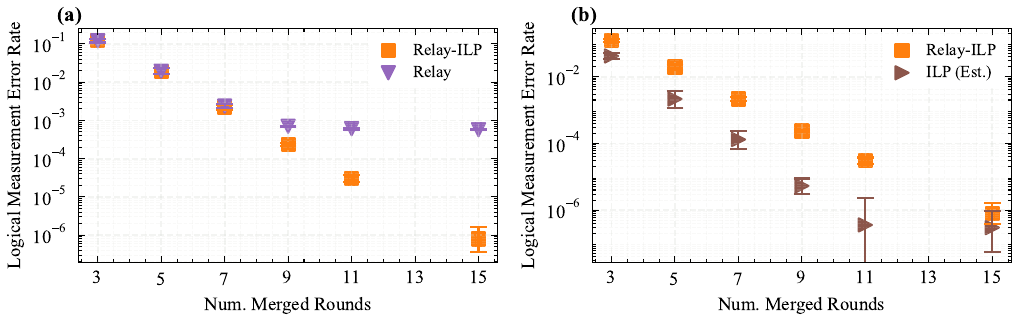}
    \caption{\textbf{(a)} Comparing the logical measurement error rate of the two-stage decoder with just Relay-BP. \textbf{(b)} Comparing the logical measurement error rate of the two-stage decoder with an estimate of an ILP decoder.
    For both panels, error bars indicate \(95\%\) Agresti-Coull confidence intervals.
    }
    \label{fig:extractor-results-detailed}
\end{figure}

\begin{table}[]
    \centering
    \begin{tabular}{cc}
    \toprule
       Parameter  & Value \\
       \midrule
        \texttt{gamma0} & \(0.35\) \\
        \texttt{pre\_iter} & \(200\) \\
        \texttt{num\_sets} & \(10\) \\
        \texttt{set\_max\_iter} & \(200\) \\
        \texttt{stop\_nconv} & 1 \\
        \texttt{gamma\_dist\_interval} & \((0.2, 0.6)\) \\
        \bottomrule
    \end{tabular}
    \caption{Relay-BP parameters used for decoding. The \texttt{gamma0} and \texttt{gamma\_dist\_interval} parameters were optimized via a small grid search on a memory experiment circuit.}
    \label{tab:relay-params}
\end{table}

\section{Comparison Table Details}
\label{sec:table-details}
In this section, we detail the assumptions made to obtain the values listed in Table~\ref{tab:comparison}.

For the surface code (fast) row, we assume the ``fast block'' layout from~\cite{litinski2019game}, which requires \(2k + \sqrt{8k} + 1\) surface patches to store \(k\) logical qubits; our \(n_{\mathrm{tot}}\) values ignore the subleading term and assume a fixed \(2\times\) overhead.
For the surface code (compact) row, we assume the ``compact block'' layout from~\cite{litinski2019game}, with one empty patch for every two data patches, and that measurements involving \(Y\) terms use twist-free surgery~\cite{chamberland2022universal}: with ancilla states prepared in advance, this requires \(d\) rounds for an \(X\)-type measurement, up to \(6\cdot d\) rounds for patch rotation, and another \(d\) rounds for a \(Z\)-type measurement.

For the systems from~\cite{yoder2025tour}, the listed compilation overhead is the maximum number of native measurements required to measure a logical operator on eleven of the twelve encoded qubits; the average is around \(18\).
For the systems from \cite{cain2026shorsalgorithm}, we include only those designed for the processor codes. Unlike in \cite{cain2026shorsalgorithm}, where qubits are allocated for checks in only one basis, our \(n_{\mathrm{tot}}\) allocates qubits for both \(X\) and \(Z\) checks, and refers to an ancilla system sized for the largest \(X\)-type logical operators sampled in \(10^5\) trials; non-CSS measurements would likely require larger systems.
For BB codes, we assume the standard overcomplete check basis~\cite{bravyi2024high-threshold2}.

Finally, for this work, \(n_{\mathrm{tot}}\) includes the ancilla qubits used to measure high-degree checks via Bell checks.

\end{document}